\documentclass[sn-mathphys-num]{sn-jnl}
\usepackage{a4wide}
\usepackage{amssymb}
\usepackage{amsthm}
\usepackage{amsfonts}
\usepackage{color}
\usepackage{graphicx}
\usepackage{graphicx,epstopdf}
\usepackage{mathtools}
\usepackage{mathrsfs}
\usepackage{multirow}
\usepackage{float}
\usepackage{rotating}
\usepackage{enumerate}
\usepackage{epsfig}
\usepackage{setspace}
\usepackage{ragged2e}
\usepackage{blindtext}
\usepackage{lineno}
\usepackage{url}
\usepackage{caption}
\usepackage{subcaption}
\usepackage{physics}
\usepackage{lscape}
\usepackage{float}
\newtheorem{theorem}{Theorem}[section]

\newtheorem{lemma}[theorem]{Lemma}

\newtheorem{remark}{Remark}
\newtheorem{proposition}[theorem]{Proposition}

\numberwithin{equation}{section}

\numberwithin{equation}{section}
\newtheorem{Thm}{Theorem}
\newtheorem{Lem}{Lemma}

\def\({\left(}
\def\){\right)}

\begin{document}
\title[]{An ``adaptive" approach to control explosive aphid populations}

\author[1]{\fnm{Aniket} \sur{ Banerjee}}\email{aniket.banerjee@sorbonne-universite.fr}

\author[2]{\fnm{Urvashi} \sur{Verma}}\email{uverma@iastate.edu}

\author[3]{\fnm{Satyam} \sur{Narayan Srivastava}}\email{srivastava@math.muni.cz}

\author[2]{\fnm{Rana} \sur{D. Parshad}}\email{rparshad@iastate.edu}

\affil[1]{Sorbonne Université, Université Paris Cité, CNRS, Laboratoire Jacques-Louis Lions, LJLL, F-75005 Paris, France}

\affil[2]{Department of Mathematics, Iowa State University, Ames, IA 50011, USA}

\affil[3]{Department of Mathematics and Statistics, Faculty of Science, Masaryk University, Brno, Czech Republic}

\abstract
{Classical models of aphid population dynamics are unable to explain multi-peak patterns in field populations. We consider the variable carrying capacity model (VCM), which can generate such complex multi-peak dynamics, but is also demonstrated to show finite-time blow-up behavior via a sign switching structural instability. We build an adaptive behavioral model with a density-dependent switch to stabilize growth, effectively eliminating blow-up, and also capable of generating multiple peaks. Furthermore, guided by empirical work on environment drivers for pests, we devise a non-autonomous model with time-dependent host plant fitness, successfully connecting transient population dynamics with abiotic drivers such as flooding. Finally, we discuss the practical significance of the results through the Economic Threshold (ET) and Economic Injury Level (EIL) calculation for all models. Our simulations all clearly show that aphid abundances exceed these threshold levels, and control is required. Our work provides a stable and biologically relevant prediction scheme for pest outbreaks and their management strategy.}


\maketitle
\section{Introduction}

Aphids are a destructive insect pest causing serious economic loss on valuable crops and inflicting damage to staple crops such as soybeans \cite{catangui2009soybean, kindlmann2010modelling}. Devising efficient resistance management practices to manage aphid populations has been the prime concern of investigations in agricultural entomology \cite{Tabashnik2013, dean2021developing}. Numerous control techniques have been suggested, such as classical refuge strategies, top-down control by predators \cite{r1}, as well as novel refuge strategies \cite{toab218}. Yet, the success of such measures relies on the strong knowledge of aphid population dynamics complicated by their special reproductive plasticity. For example, the soybean aphid \emph{(Aphis glycines)} experiences explosive growth in the asexual stage during summer months, reaching populations of several thousand per plant, only to collapse in late summer as host plants begin to senesce, and winged aphids subsequently migrate to wintering hosts \cite{ragsdale2004soybean, ragsdale2011ecology, tilmon2011biology}.

Although current models fare well in representing individual boom-bust cycles in aphid populations, long-term field data projects a more nuanced picture \cite{costamagna2007exponential, matis2007stochastic,matis2009population}. In particular, in real field data, we frequently find more than one population peak per growing season - a process less easily explained \cite{K07, D90}. New advances in ecological modeling indicate that such multi-peaked population dynamics could be modeled as solutions of a variable carrying capacity model. In contrast to the static models of conventional approaches that address resources as such, this method appreciates how constantly the host plant quality varies. Conditions such as changing nitrogen availability, induced plant defenses, and environmental pressures all make the aphid's food source ``variable".
In the current manuscript, we propose models that can effectively capture this inherent variability. By coupling the population dynamics of aphids and their fluctuating resources, we reproduce the multi-peak patterns that are typical of field data. Our models also provide us with precise predictions—and ultimately, more effective pest management practices.

Models for aphid dynamics exhibiting multiple population peaks within a single growing season have been recently proposed by Kindlmann et al. \cite{DD20}. These consider a variable carrying capacity. Note, although this model captures multiple peaks, it also exhibits unbounded growth, which is biologically unrealistic, since in reality, resources are always limited. Thus, a different modeling technique needs to be delved into to address the multiple peak dynamics (possibly due to abiotic stressors), but without unbounded growth. One such method is the adaptive behavior model that has been investigated in the past for pests on grapevines  \cite{ainseba2023}, where the dynamical behavior of the pest is assumed to evolve over the season, based on changes in abiotic factors or seasonal diapause. Inspired by such and related works, in the current work, we have formulated an adaptive behavior model to include evolving pest dynamics over time, due to a population boom.

Transient population dynamics have been increasingly recognized as a central theme in ecology, shaping population persistence, species interactions, and management outcomes; some notable research in this direction includes \cite{t1,t2,t3,t4}.
Classical autonomous models for pest dynamics, which assume constant parameters \cite{toab218}, often fail to capture the transient dynamics of ecological systems driven by changing environmental factors, such as flooding, drought, prolonged daylight, and temperature variation. These can, however, be effectively addressed by non-autonomous systems with time-dependent parameters \cite{coutinhoa2006threshold,hallam1981non,baranyi1993non}. A plant's fitness, for example, has been observed to change due to environmental factors such as excessive water stress, caused due to flooding events \cite{lewis2025host}. Motivated by these observations, we consider a non-autonomous model that incorporates time-dependent plant fitness and pest growth rate, thereby allowing us to explore transient periodic behavior and eventual extinction dynamics that cannot be captured by autonomous counterparts.

\par Effective pest management relies on maintaining pest populations within acceptable limits, which is guided by the ET (Economic Threshold) and EIL (Economic Injury Level).  ET represents the pest population level at
which a control action has to be taken in order to prevent economic
harm to a crop, while EIL indicates the population level at which pests cause actual economic damage to crops. Ragsdale et al. \cite{ragsdale2007economic} were the first to report the ET and EIL for soybean aphids.  Closely monitoring these two key measures across all the models helps us to compare whether these different methods are able to keep pest populations in check. Thus, simulations with the models we develop take into account these key metrics, and several results are provided that show the various model settings that can cause delay to reach ET and therefore EIL levels.

The paper is organized as follows: In Section \ref{Aphid Population Dynamics}, we review the Kindlmann model with varying carrying capacity and analyze its blow-up behavior. Section \ref{Adaptive behavior model} introduces the adaptive behavior model and shows the well-posedness of the model. Section \ref{Non-autonomous Model} demonstrates a realistic non-autonomous model describing the change of fitness of aphids over time, and Section \ref{ET_EIL} discusses the economic damage caused by pest population and population threshold to perform control strategies. We conclude with a summary and future research directions in Section \ref{Discussion And Conclusion}.

\section{Aphid Population Dynamics}
\label{Aphid Population Dynamics}
\subsection{Model for boom-bust dynamics}

Kindlmann and co-authors \cite{kindlmann2010modelling} introduce the following population model to capture the boom-bust dynamics commonly seen in aphid populations,

 \begin{equation}
\label{eq:cl1}
\begin{aligned}
& \frac{dh}{dt} = a x,  \ h(0)=0 \\
& \frac{dx}{dt} =  (r-h)x,  \ x(0)=x_{0}
\end{aligned}
\end{equation}

%\begin{equation}
%\label{eq:cl2}
%\frac{dx}{dt} =  (r-h)x,  \ x(0)=x_{0}
%\end{equation}

Here $h (t)$ is the cumulative population density of a single aphid biotype at time $t$; $x(t)$ is the population density at time $t$, $a$ is a scalar constant, and $r$ is the growth rate of the aphids. The aphid population initially rises, driven by the linear growth term - this is the "boom" phase. Still, as the cumulative density exceeds the growth rate $r$, the population decreases, driven by competition. This results in the "bust" phase \cite{kot2001elements} and "boom-bust" phases, which can be seen in fig \ref{fig:kcm_boombust}. This is commonly seen in the population dynamics of aphids and has been observed in soybean aphids in North America \cite{catangui2009soybean}, with colonization in June and then gradual buildup of population, peaking in August and declining, with aphids dispersing in September to their overwintering host.
These dynamics differ from those predicted by the classical logistic growth model, which shows convergence to the carrying capacity state.  This also provides a framework to investigate additional situations, such as if the carrying capacity were variable in time.

\subsection{Logistic model with variable carrying
capacity and growth rate affected by
cumulative density}

It has been observed in the field on occasion that the classical (one-peak) boom-bust dynamics will not occur; rather, one may see multiple peaks during the growing season, one initially followed by a downturn and then another peak \cite{D90}. This could be for a plethora of reasons. Host plant suitability changes over the growing season, environmental and/or weather-driven changes, an excessive abundance of aphids at the beginning of the season, enhanced competition, and reduced fecundity \cite{K07, Ke50}.
The following population model for aphids \cite{K07}, with variable carrying capacity, is introduced to capture such dynamics. We abbreviate this as VCM (variable carrying capacity logistic model) henceforth,
\begin{equation}
\begin{split}
\label{eq:1}
\frac{dh}{dt} &= a x, \ h(0) =0 \\
\frac{dx}{dt} &=  (r-h)x\left(1-\frac{x}{k}\right),  \ x(0) =x_{0}
\end{split}
\end{equation}

\begin{equation}
\label{eq:22a}
\text{Here, } \ 
k=k(t)=\left( k_{max} - k_{min}\right)f(t) + k_{min}, \
f(t) = \left(\frac{cos( d \pi t)+1}{2}\right)
\end{equation}

$f(t)$ is modeled as a trigonometric function to mimic the environment/resources fluctuating, such that the total carrying capacity fluctuates between $k_{max}$ and $k_{min}$ values (as $f(t)$ fluctuates between $0$ and $1$).

%\begin{figure}[H]%simulated in Matlab code
%\centering
%\includegraphics[width = 11cm]{klm_x0_10_t_80.eps }
%\caption{The time series figure shows the carrying capacity and aphid population density with $t = 80$ for the Logistic population model with variable carrying capacity. The parameter set used: $ K_{max} = 10000, K_{min} =1, d =.033, a = 0.000005, r=0.3$. The initial aphid population density $x_0 = 10$.}
%need to explain param and fig
%param used K_max = 10000; K_min =1; d =.033;a = 0.000005; r=0.3;x_0=50
%\label{fig:klm_x0_80}
%\end{figure}

%\begin{figure}[H]%simulated in Matlab code
%\centering
%\includegraphics[width = 11cm]{klm_x0_10_t_85.eps }
%\caption{The time series figure shows the carrying capacity and finite time blow-up of the aphid population density at $t \approx 85 $ for the Logistic population model with variable carrying capacity. The parameter set used: $ K_{max} = 10000, K_{min} =1, d =.033, a = 0.000005, r=0.3$. The initial aphid population density $x_0 = 10$.}
%need to explain param and fig
%param used K_max = 10000; K_min =1; d =.033;a = 0.000005; r=0.3;x_0=50
%\label{fig:klm_x0_85}
%\end{figure}

\begin{figure}%simulated in Matlab code
\begin{subfigure}{.28\textwidth}
    \includegraphics[width = 4.8cm,height=5.0cm]{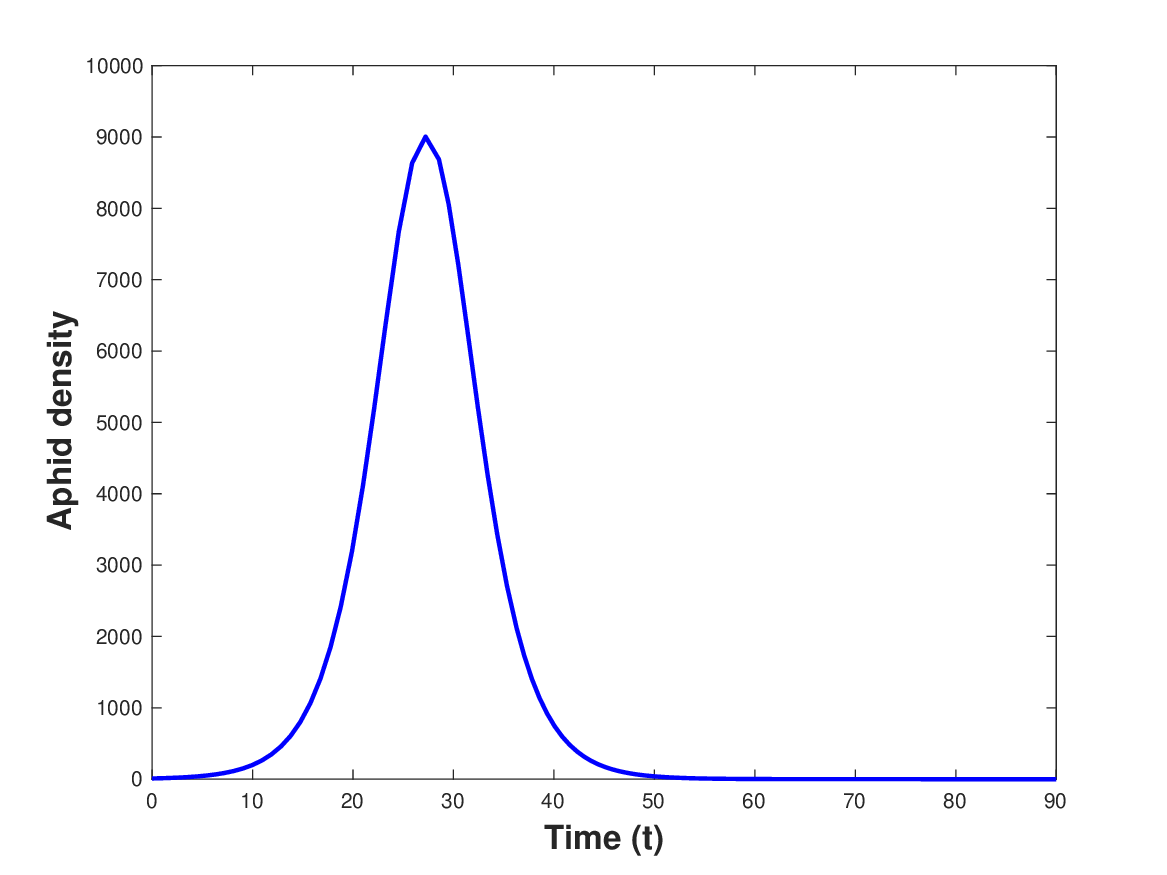}
    \subcaption{}
\label{fig:kcm_boombust}
\end{subfigure}
\begin{subfigure}{.34\textwidth}
\centering
\includegraphics[width= 5.8cm, height=5cm]{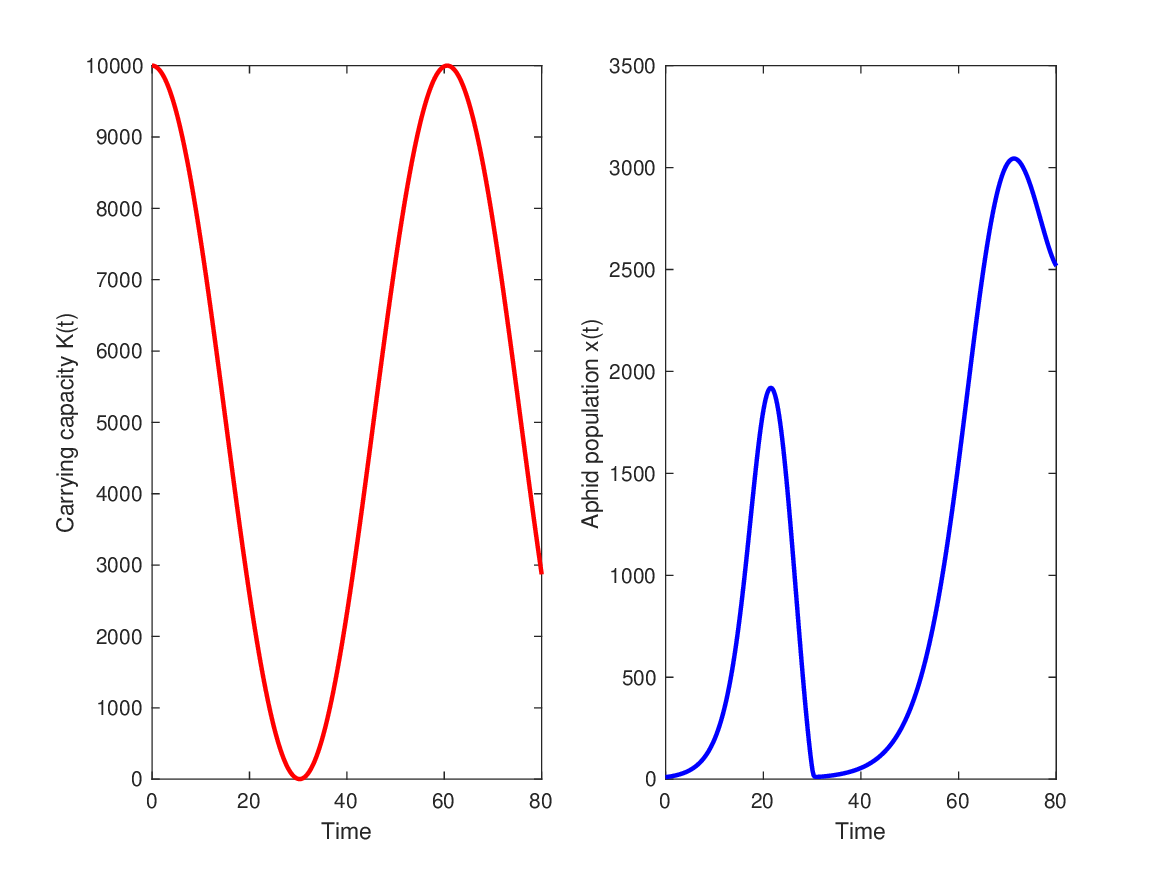 }
\subcaption{}
%need to explain param and fig
%param used K_max = 10000; K_min =1; d =.033;a = 0.000005; r=0.3;x_0=50
\label{fig:klm_x0_80}
\end{subfigure}
\begin{subfigure}{.34\textwidth}
\centering
\includegraphics[width= 5.8cm, height=5cm]{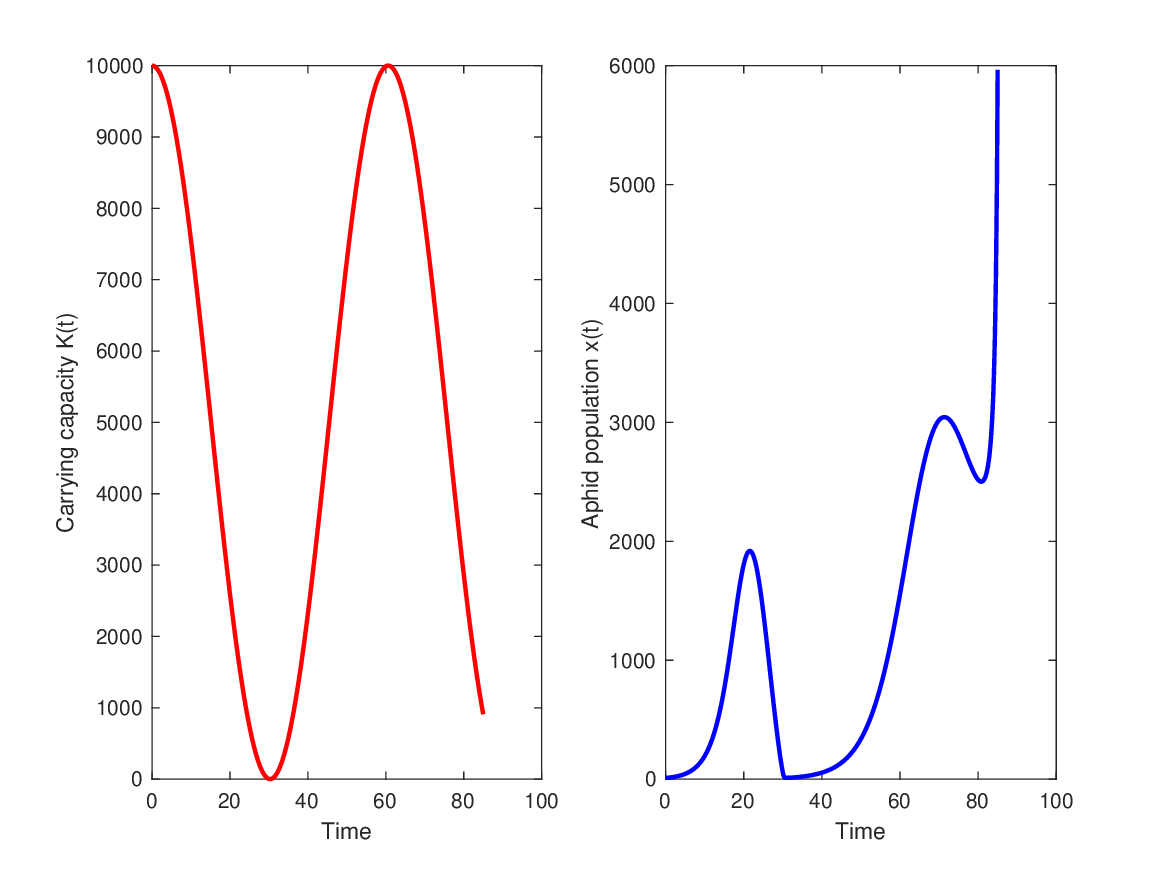 }
\subcaption{}
\label{fig:klm_x0_85}
\end{subfigure}
\caption{For Figure \ref{fig:kcm_boombust}: The time series figure shows the boom-bust scenario for model \eqref{eq:cl1}. The parameter set used:
$ a = 0.000005, r=0.3$. The initial population density is $h_0=0,x_0=10$. For Figure  \ref{fig:klm_x0_80}:  The time series figure shows the carrying capacity and aphid population density with $t = 80$ for the Logistic population model with variable carrying capacity \eqref{eq:1}. The parameter set used: $ K_{max} = 10000, K_{min} =1, d =.033, a = 0.000005, r=0.3$. The initial aphid population density is $x_0 = 10$. For Figure \ref{fig:klm_x0_85}: The time series figure shows the carrying capacity and finite time blow-up of the aphid population density at $t \approx 85 $ for the Logistic population model with variable carrying capacity \eqref{eq:1}. The parameter set used: $ K_{max} = 10000, K_{min} =1, d =.033, a = 0.000005, r=0.3$. The initial aphid population density is $x_0 = 10$.}
\end{figure}

\begin{table}
\centering
 \begin{tabular}{ |p{2.3cm}||p{4.4cm}||p{2.2cm}|  }
 \hline
 \multicolumn{3}{ |c| }{Blow-up time with  variation in $``a"$} \\
 \hline
 Value of $``a"$ & Blow-up for initial data $x_0$ & Time (days)  \\
 \hline
 0.000005  & 105 & $\approx $ 30\\
 \hline
 0.00005 & 10  & $\approx $ 27 \\ %(even initial data as 1 will blow up on 30th day)
  \hline
 0.0005 & 1 & $\approx $ 30\\ %(3 is not blowing up)
  \hline
 0.005  & 9999  & $\approx $ 1 \\%(may remove this)
  \hline
\end{tabular}
\caption{A table to show the sensitivity of system \eqref{eq:1} to the parameter $a$. The table shows for different values of the solution blow-up and the time, and the lowest initial data when the solution blow-up occurs.}
\label{table_blowup}
\end{table}

\begin{remark}
We see via Fig. \ref{fig:klm_x0_80} that the VCM can predict multiple peaks in a single season, one early on around day 20, and one later at around day 75 - however, running the simulation a little longer yields an ``exploding" solution around day 85, as seen via Fig. \ref{fig:klm_x0_85}.
    \end{remark}
Thus, the simulation results, as seen via Fig. \ref{fig:klm_x0_85}, motivate us to rigorously investigate both analytically and numerically the blow-up dynamic present in the VCM \eqref{eq:1}.

\begin{remark}
Heuristically, the blow-up in the VCM \eqref{eq:1} happens when the cumulative density $h$ exceeds r, and the sign of $(r-h)$ is negative. Thus, the standard logistic equation (or even logistic equation with variable carrying capacity $k$) is typically described via $x - \frac{x^{2}}{k}$, flips the sign to have a term like
$\frac{x^{2}}{k} - x$, which can blow-up for sufficiently chosen initial conditions. Similar blow-up results due to sign-changing non-linearity can be seen in works \cite{S47, T22a}.
    \end{remark}

\subsubsection{Some standard and auxiliary results}
\begin{lemma}
Consider the VCM given by \eqref{eq:1}. Then for positive initial data  $x_{0} > 0$, solutions to 
\eqref{eq:1} remain non-negative.
\label{lemma:VNLM_1}
\end{lemma}

\begin{proof}
The result follows by checking the quasi-positivity condition via Lemma \ref{lem:qp}.
\end{proof}
\begin{lemma}
\label{lem:l1}
Consider the VCM given by \eqref{eq:1}. Then $\forall \epsilon > 0$, $\exists \delta(\epsilon) > 0, x^{*}_{0}(\epsilon) > 0$ s.t.

\begin{equation}
x(t) > k_{max}\left( \frac{r}{a k_{min}}  + 1 \right) + \epsilon, \ \forall t \in [0,\delta],
\end{equation}

for all solutions to \eqref{eq:1}, initiating from the initial data $x^{*}_{0}(\epsilon) > 0$.
\end{lemma}

\begin{proof}
The result follows via continuity of solutions to \eqref{eq:1} via Theorem \ref{thm:class3}.
\end{proof}

\subsubsection{Blow-up in finite time in VCM }
\begin{theorem}
\label{thm:t1t}
Consider the VCM given by \eqref{eq:1}. Then, for sufficiently large initial data, solutions to
\eqref{eq:1}-\eqref{eq:22a} blow-up in finite time, that is, $
\limsup_{t \rightarrow T^{*} < \infty} |x(t)| \rightarrow + \infty.$
\end{theorem}
\begin{proof}
See Appendix \ref{proof_thm:t1t}.
\end{proof}
\begin{theorem}
\label{thm:t1}
Consider the VCM given by \eqref{eq:1}. Then for initial data sufficiently large, that is $x_{0} > k_{max}\left( \frac{r}{a k_{min}}  + 1 \right) $, solutions to 
\eqref{eq:1} blow-up in finite time, that is, $ \limsup_{t \rightarrow T^{*} < \infty} |x(t)| \rightarrow + \infty.$
\end{theorem}
\begin{proof}
    See Appendix \ref{proof_thm:t1}.
\end{proof}

\begin{remark}
The blow-up in $h(t)$, follows using \eqref{eq:3}. In the more general case, one can compare to the ODE, $\boxed{y^{'} = C_{3}y^{p} - C_{4}y^{q}}, p>q>1, C_{3}>0, C_{4} > 0$.
\end{remark}

\begin{lemma}
\label{Lem:l12}
Consider the VCM given by \eqref{eq:1}. Then for initial data sufficiently large, i.e., $x_{0} > k_{max}\left( \frac{r}{a k_{min}}  + 1 \right) $, the cumulative pest density blows up in finite time, i.e., $\limsup_{t \rightarrow T^{*} < \infty} |h(t)| \rightarrow + \infty.$
\end{lemma}

\begin{proof}
See Appendix \ref{proof_Lem:l12}.
\end{proof}

\subsubsection{Blow-up for other initial conditions}

We next explore the case when blow-up is possible for other positive initial data, possibly small. 

\begin{remark}
The estimate via Theorem \ref{thm:t1}, is only sufficient, in that if the initial data is large enough, $x_{0} > k_{max}\left( \frac{r}{a k_{min}}  + 1 \right) $, then blow-up will occur. This threshold depends strongly on the parameter $a$. 
\end{remark}

The smaller $a$ is, the larger the data required for blow-up. However, this is not seen in simulations. Rather small $a$ leads to blow-up for essentially any positive initial condition. This motivates proving blow-up for any positive initial data under certain parametric restrictions. One approach is to ``construct" a lower solution and derive conditions under which this lower solution blows up for any initial condition. An ``approximate" ODE for the lower estimate is given by,
\begin{equation}
\frac{dx}{dt} = \frac{a}{k_{max}}x^{3} - \left( \frac{r}{k_{min}}  + a \right) x^2 + rx
\end{equation}
We note this is only a crude approximation because (without enforcing any positivity conditions) we require,
$ h(t) = \int^{t}_{0}a x(s)ds \approx x(t),
$
If one applies this ``approximation" and explores the cubic,
\begin{equation}
\label{eq:c3}
\left(  \frac{a}{k_{max}}\right)x^{3} - \left( \frac{r}{ k_{min}}  + a \right) x^2 + rx = x \left(  \left( \frac{a}{k_{max}}\right)x^{2} - \left( \frac{r}{ k_{min}}  + a \right) x + r \right)
\end{equation}

This has one real zero, $x=0$, for others we check the discriminant. 
Then, the other two roots must be complex. Since the coefficient of the leading term in the cubic is positive, it must approach positive infinity as $x \rightarrow \infty$. Thus, standard phase analysis will yield blow-up for any positive initial data. The required condition for this is,
\begin{equation}
\left( \frac{r}{ k_{min}}  + a \right) ^{2} - 4 \left(\frac{a}{k_{max}}\right) r < 0
\end{equation}
However, this is never true as $k_{max} > k_{min}$. Thus, there exist two positive roots, which yield blow-up in finite time for the cubic ODE \eqref{eq:c3}, see figure \eqref{plot_poly}. So, we investigate different modeling techniques to formulate a bounded model with the same dynamics of multiple peaks as in VCM \eqref{eq:1}. A very effective strategy that has been studied extensively is the delay equation models.

\subsubsection{Blow-up via delay equations}
We recap the following result from the literature, \cite{Ez06}

\begin{proposition}
Consider the following delayed equation
\begin{equation}
\label{eq:d13}
\begin{aligned}
& y^{'}(t) = |y(t)|^{p} - |y(t-\tau)|^{q}, \ t>0,  \\
& y(\tau_{1}) = \phi(\tau_{1}), \ \tau_{1} \in [-\tau,0]
\end{aligned}
    \end{equation}

Assume that $p > \max{(q,1)}$ and $\phi$ satisfying $\phi(0) \geq |\phi(t)|^ {\frac{q}{p}}$ for all $t \in [- \tau,0], \phi(0) \geq  1, \phi(0) > |\phi(-\tau)|^{\frac{q}{p}}$, then the solution \eqref{eq:d13} blows-up in finite time.
\end{proposition}
In order to proceed, we need to modify the above and state the following Theorem,
\begin{theorem}
\label{thm:l1k}
    Consider the following delayed equation
\begin{equation}
\label{eq:d1}
\begin{aligned}
& y^{'}(t) = |y(t)|^{p} - M|y(t-\tau)|^{q}, \ t>0, \\
& y(\tau_{1}) = \phi(\tau_{1}), \ \tau_{1} \in [-\tau,0]. 
\end{aligned}
    \end{equation}
Assume that $p > \max{(q,1)}$ and $\phi$ satisfying $\phi(0)  \geq M^{\frac{1}{p}}|\phi(t)|^{\frac{q}{p}}$ for all $t \in [- \tau,0], \phi(0) \geq M, \phi(0) > M^{\frac{1}{p}}|\phi(-\tau)|^{\frac{q}{p}}$, then the solution \eqref{eq:d1} blows-up in finite time.
\end{theorem}

\begin{proof}
   
The proof follows ideas in \cite{Ez06}. See Appendix \ref{proof_thm:l1k}.
\end{proof}
\begin{theorem}
\label{thm:t13}
Consider the VCM model \eqref{eq:1}. Then for initial data sufficiently large, that is $x_{0} > k_{max}\left( \frac{r}{a k_{min}}  + 1 \right)\frac{T^{2}}{a^{2}} $, solutions to 
\eqref{eq:1} blow-up in finite time, i.e., $\limsup_{t \rightarrow T^{*} < \infty} |x(t)| \rightarrow + \infty.$

\end{theorem}
\begin{proof}
    See Appendix \ref{proof_thm:t13}.
\end{proof}

\begin{remark}
As seen earlier, system \eqref{eq:1} shows the solution can blow up, and we look into the sensitivity of the solution with the parameter $a$. We investigate different scenarios to understand the blow-up in the solution. We study it in two different ways- first, we investigate the change in solution behavior with changing the initial condition, see figure \ref{fig:consta}, and secondly, we fix everything and change the parameter $a$, see figure \ref{fig:variable_a}. 
\end{remark}

%change the layout of figure make them side by side 
\begin{figure}
%\resizebox{0.85\textwidth}{!}{
%\begin{minipage}{\textwidth}
\begin{subfigure}{.32\textwidth}
  \centering
  \includegraphics[width= 5.4cm, height=5cm]{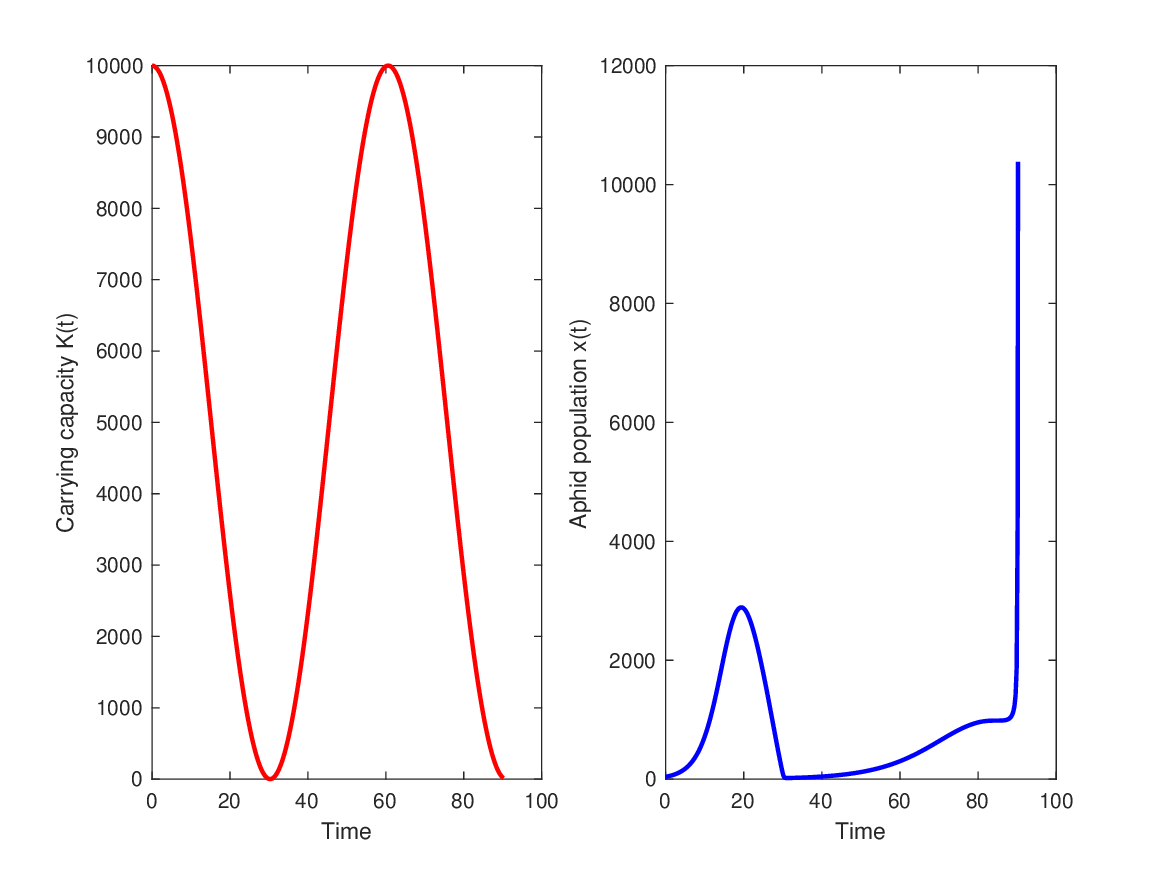}
  \caption{}
  \label{fig:sfig1}
\end{subfigure}%
\begin{subfigure}{.32\textwidth}
  \centering
  \includegraphics[width=5.4cm, height=5cm]{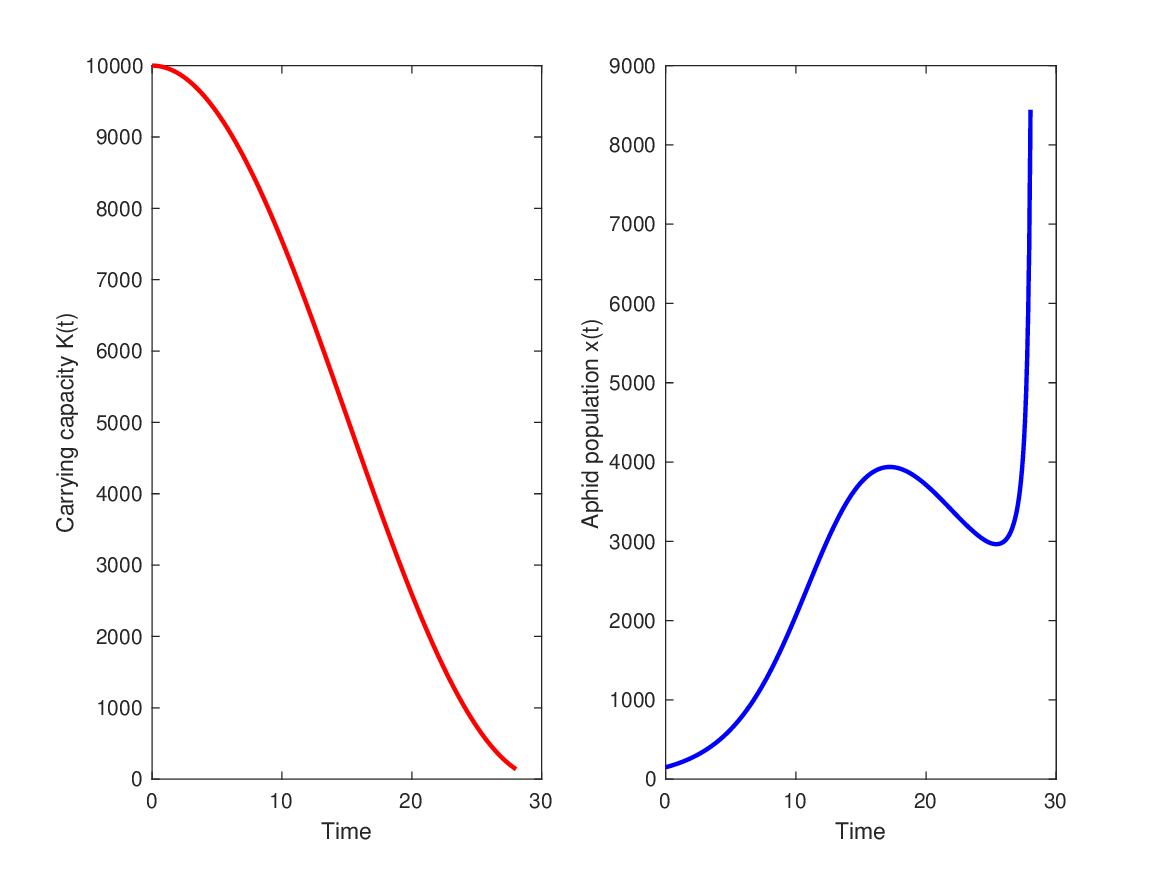}
  \caption{}
  \label{fig:sfig2}
\end{subfigure}
\begin{subfigure}{.32\textwidth}
  \centering
  \includegraphics[width= 5.4cm, height=5cm]{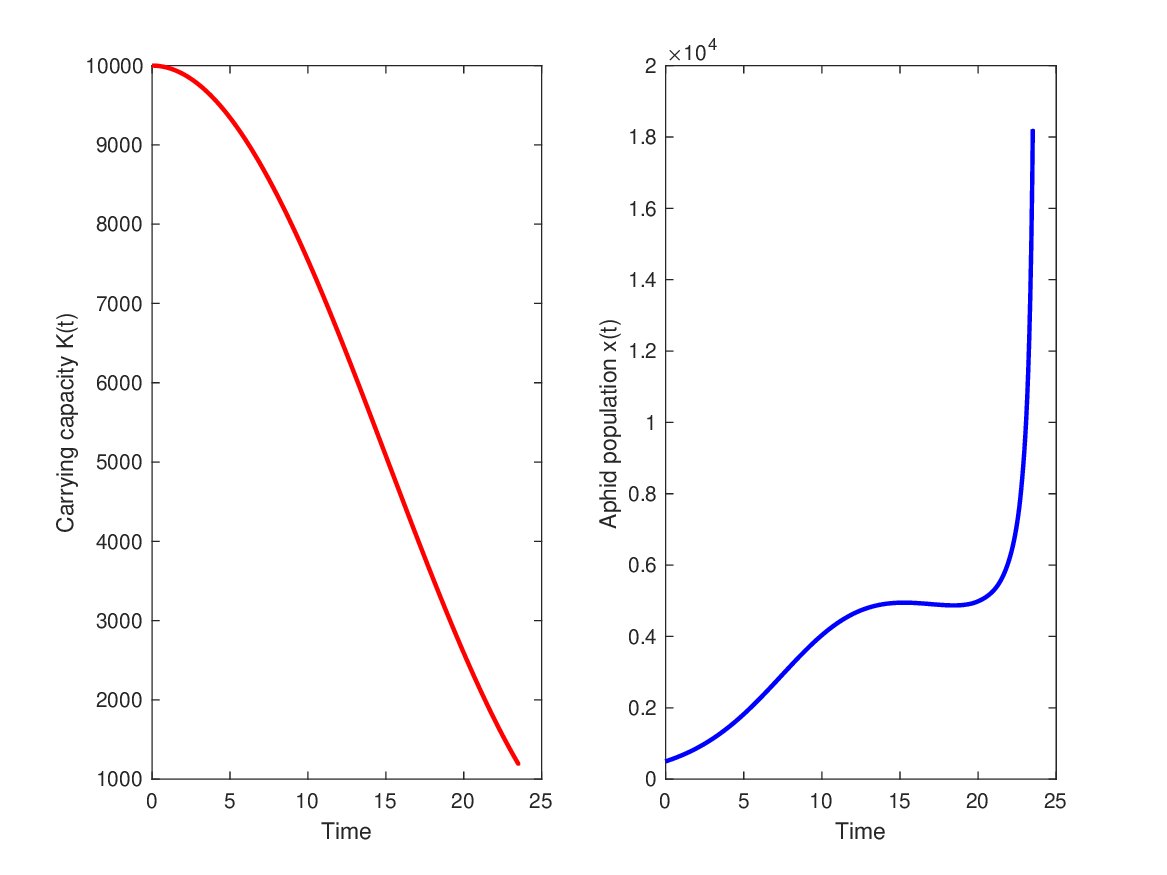}
  \caption{}
  \label{fig:sfig3}
\end{subfigure}
%\begin{subfigure}{.55\textwidth}
 % \centering
  %\includegraphics[width= 8cm, height=5cm]{blowup_x0_1000_t_21.eps}
  %\caption{}
  %\label{fig:sfig4}
%\end{subfigure}
%\end{minipage}
%}
\caption{The time series figure shows the carrying capacity and finite time blow-up in the aphid population with variation in the initial population. The parameter set: $ K_{max} = 10000, K_{min} =1, d =.033, a = 0.000005, r=0.3$. The aphid population blows up at: Figure \ref{fig:sfig1} $x_0 = 40$, $t \approx 85$ with $x \approx 2000$, Figure  \ref{fig:sfig2} $x_0 = 150$, $t \approx 28$ with $x \approx 4800 $, Figure  \ref{fig:sfig3} $x_0 = 500$, $t \approx 23 $ with $x \approx 6800 $. %\ref{fig:sfig4})  $x_0 = 1000$, $t \approx 21 $ with $x \approx 7800$. 
}
\label{fig:consta}
\end{figure}

\begin{figure}
%\resizebox{0.85\textwidth}{!}{
%\begin{minipage}{\textwidth}
%figure representing row 1 of table 
\begin{subfigure}{.32\textwidth}
  \centering
  \includegraphics[width= 5.4cm, height=5cm]{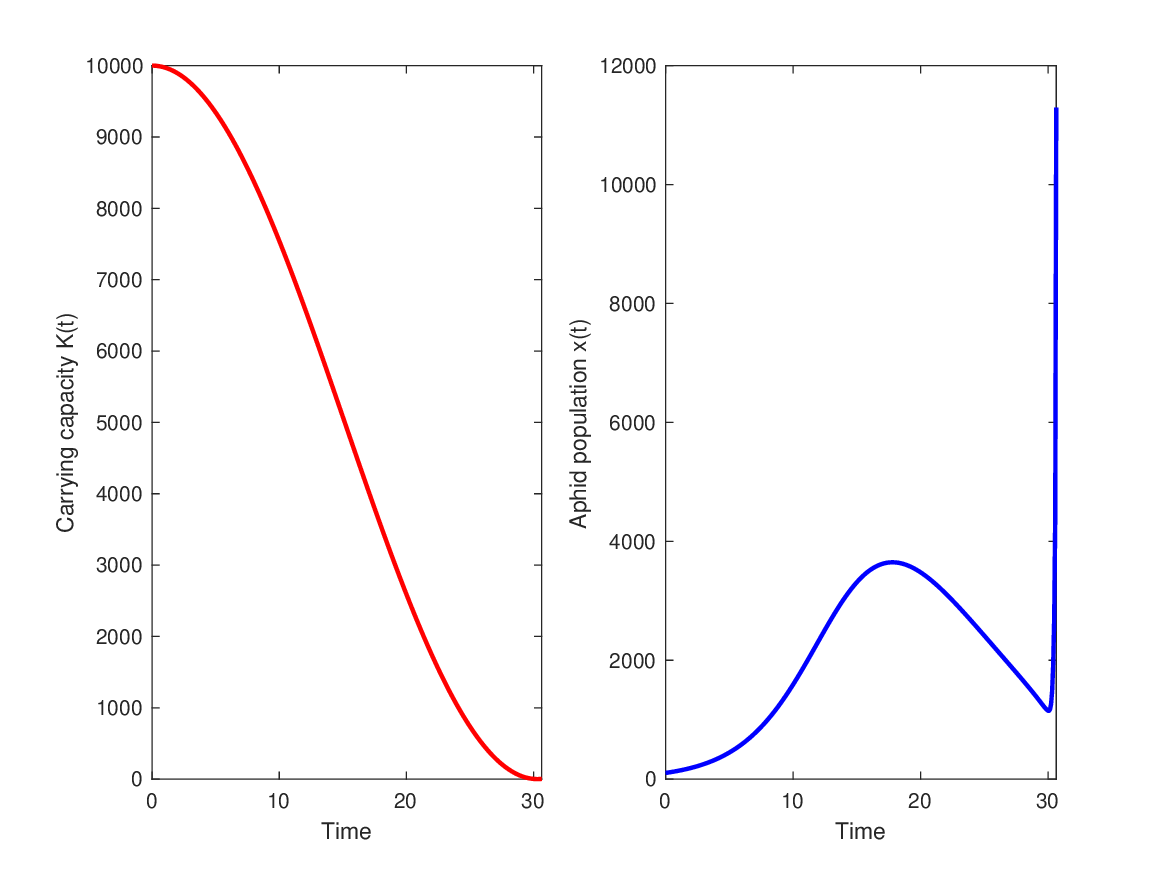}
  \caption{}
  \label{x0_105}
\end{subfigure}%
%figure representing row 2 of table 
\begin{subfigure}{.32\textwidth}
  \centering
  \includegraphics[width= 5.4cm, height=5cm]{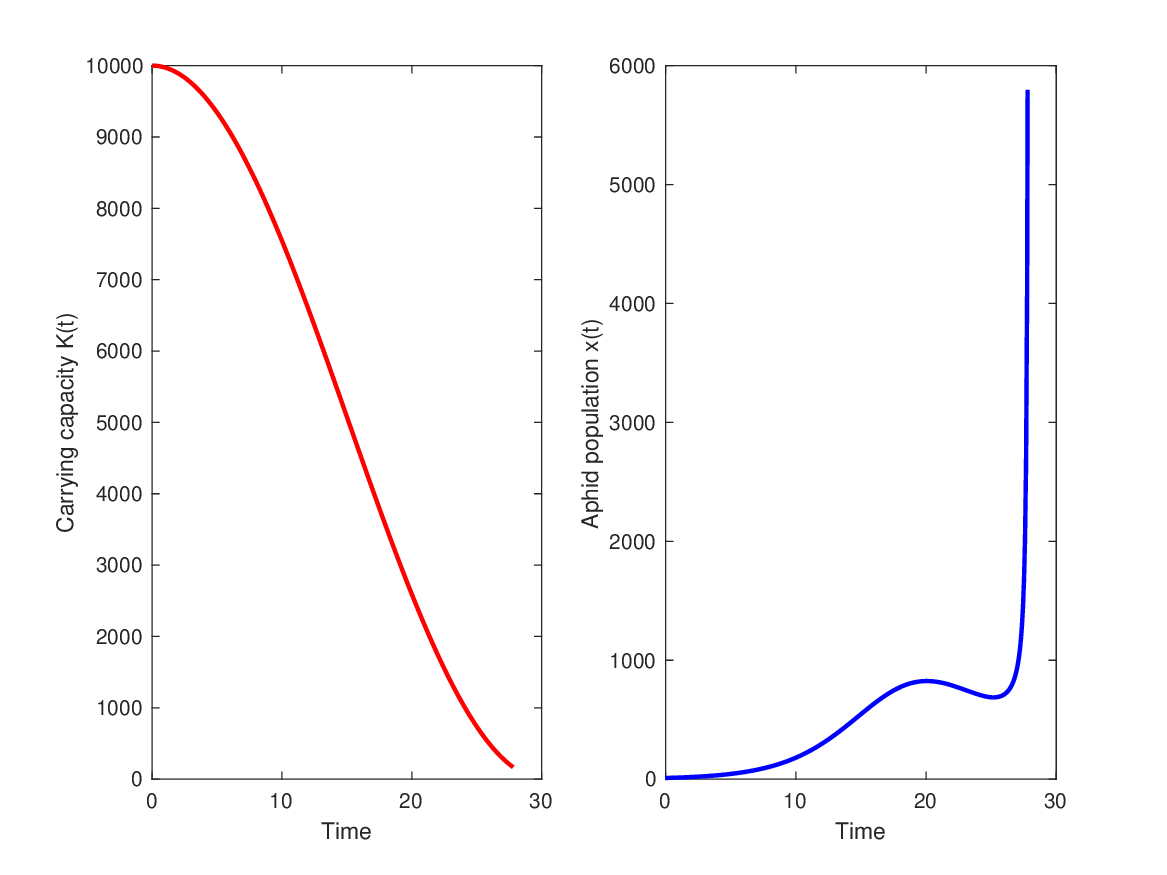}
  \caption{}
  \label{x0_10}
\end{subfigure}
%\vspace{0.5cm}
%figure representing row 3 of table 
\begin{subfigure}{.32\textwidth}
  \centering
  \includegraphics[width= 5.4cm, height=5cm]{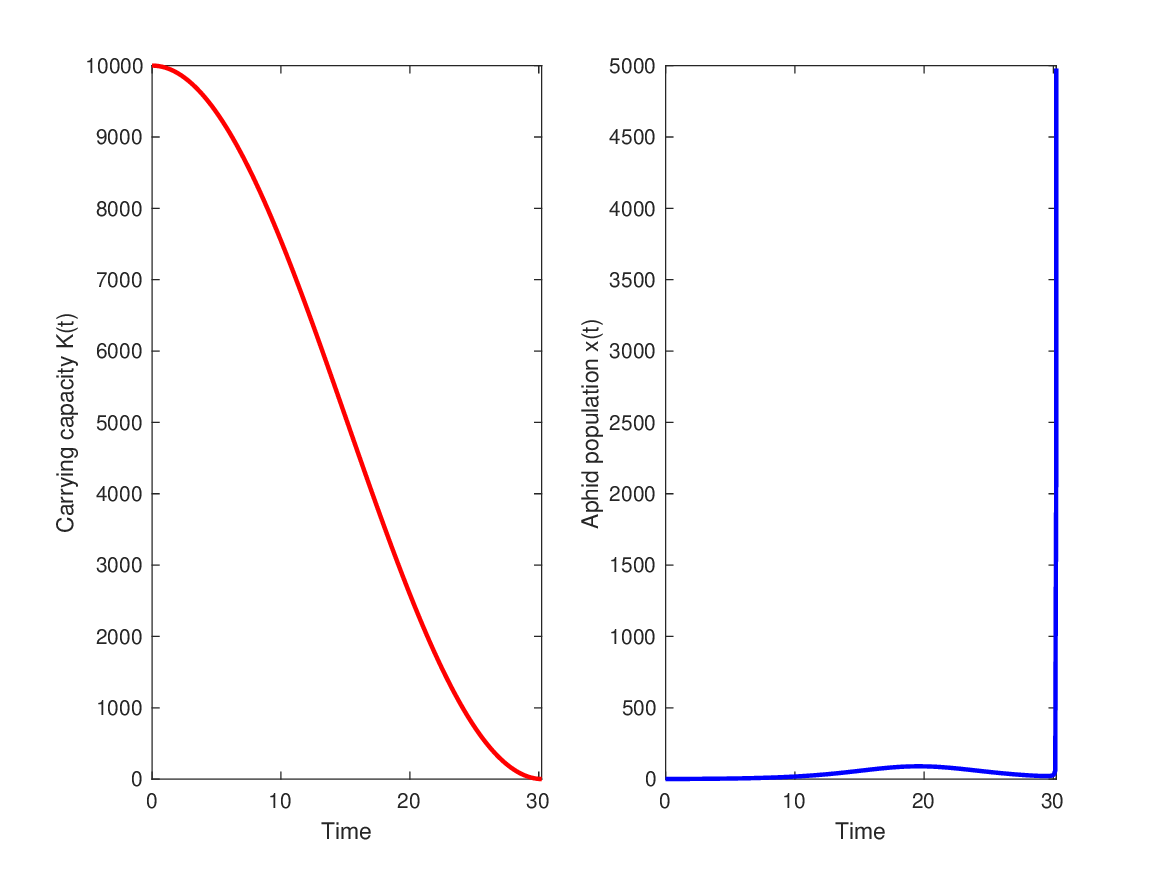}
  \caption{}
  \label{x0_1}
\end{subfigure}
%figure representing row 4 of table 
%\begin{subfigure}{.55\textwidth}
  %\centering
 % \includegraphics[width= 8cm, height=5cm]%{x0_9999_a_0.005.eps}
  %\caption{}
 % \label{x0_9999}
%\end{subfigure}
%\end{minipage}
%}
\caption{The time series figure shows the carrying capacity and finite time blow-up with variation in $``a"$. The parameter set: $ K_{max} = 10000, K_{min} =1, d =.033, r=0.3$. For Figure \ref{x0_105} $ a = 0.000005$, $x_0 = 105$ and population blows up at  $t \approx 30 $, Figure \ref{x0_10} $ a = 0.00005$, $x_0 = 10$ and population blows up at $t \approx 27 $, Figure \ref{x0_1} $a= 0.0005$, $x_0 = 1$ and population blows up at  $t \approx 30 $.  %\ref{x0_9999}) $a=0.005$, $x_0 = 9999$ and population blows up at $t \approx 1$
}
\label{fig:variable_a}
\end{figure}

\section{Adaptive behavior model}
\label{Adaptive behavior model}
In this section we propose an adaptive behavior model. Our premise is that
during the lifetime of the soybean aphid, the population is driven by variable carrying capacity as seen in model \eqref{eq:1}-\eqref{eq:22a}. However, just as the population explosion takes place at time $t=t_{bu}$ (i.e., the solution blows up starting from time $t_{bu}$) the population dynamics shifts to the classical Kindlmann model \eqref{eq:cl1}. In this sense the model is ``adaptive". The model is formulated is as follows:
\vspace{0.2cm}

a) \textbf{The variable carrying capacity dynamics period} ($0 \leq t < t_{ub}$): Before the solution blow up, the population evolves according to a variable carrying capacity $k(t)$ during the time interval $(0,t_{ub})$. The model is given by the system
\begin{equation}\label{model:part1}
    \begin{cases}
        \frac{dh}{dt} = a x \\
        \frac{dx}{dt} = (r - h) x \left(1 - \frac{x}{k(t)}\right)\\
        k(t)=(k_{max}-k_{min})\left ( \frac{\cos(d\pi t) +1}{2}\right)+k_{min}
    \end{cases}
\end{equation}

b) \textbf{The post-blowup period} ($t_{ub}\leq t\leq t_{end}$): Due to the onset of exponential growth of the population at $t=t_{ub}$, the population growth and competition dominate the effect of variable carrying capacity, and the population evolves as per classical hump-dynamics, which is described by 
\begin{equation}\label{model:part2}
    \begin{cases}
        \frac{dh}{dt} = a x \\
        \frac{dx}{dt} = (r - h) x 
    \end{cases}
\end{equation}
Where $h(t)$ is the cumulative population density of aphid biotype at time $t$; $x(t)$  is the population density at time $t$, $a$ is a scalar constant, and $r$ is the growth rate of the aphids. 

Thus it is critical to estimate the time $t=t_{ub}$, and ensure that \ref{model:part1}, does not blow up before $t=t_{ub}$. A strategy here is to derive a lower bound $T^{**}$, for the blow-up time of \ref{model:part1}, and then ensure that $t_{ub} < T^{**}$, thus yielding existence on $[0,t_{ub}]$.

The following lemma enables this,

\begin{lemma}
Consider model \ref{model:part1}, with a given initial condition $x_{0}$. Then the solution $x(t)$ will only blow up after a finite time, $T^{*} = r + \sqrt{r^{2}+\frac{2a}{(x_{0})^{2}}}$.
\end{lemma}

\begin{proof}
    By direct comparison we have,
    \begin{equation}
        \frac{dx}{dt} \leq rx + h(t) x^{2} \leq rx^{3} + (a t x)x^{2} = (r+at)x^{3}.
    \end{equation}
    This follows via monotonocity of $x$, positivity, and the fact that $h(t) \leq a t\sup_{[0,t]} x(s) = a t x(t)$.
    We can compute the blow-up time of the supersolution by solving, $\frac{dx}{dt} = (r+at)x^{3}$. Separation of variables entails,

    \begin{equation}
        -\frac{1}{(x(t))^{2}} = rt + a \frac{t^{2}}{2} -  \frac{1}{(x_{0})^{2}} 
    \end{equation}
    and,
     \begin{equation}
       x(t)= \frac{1}{\sqrt{\frac{1}{(x_{0})^{2}} - rt - a \frac{t^{2}}{2}}}
    \end{equation}
    The above blows up when $T^{*} = r + \sqrt{r^{2}+\frac{2a}{(x_{0})^{2}}}$.
    By direct comparison we have that, the blow-up time $T^{**}$ of \ref{model:part1}, is such that 
$T^{*} < T^{**}$.
    
    Thus choosing $t_{ub} < T^{*} < T^{**}$ provides the requisite lower bound on the blow-up time for \ref{model:part1}. 
\end{proof}
\begin{remark}
    From the above lemma a convenient estimate is $\boxed{t_{ub} = 2r}$ as $2r < r+\sqrt{r^{2}+\frac{2a}{(x_{0})^{2}}}$.
\end{remark}

\begin{theorem}
    System \ref{model:part1}-\ref{model:part2} is well-posed.
\end{theorem}

\begin{proof}
To show the well-posedness of the given model \ref{model:part1}-\ref{model:part2} we show the following:
 i) Existence of the solution,
ii) Uniqueness of the system and iii) Continuous dependence on initial condition and parameters. For, \textit{existence of the solution}:
Let $k(t) \in C^\infty$. The time where system \ref{model:part1} shifts to system \ref{model:part2} is defined as $t_{ub}=\inf{\{ t_n | x(t_n) \to \infty \}} $. For system \ref{model:part2} at $t \in [t_{bu},t_{end}]$, the system are polynomials on $x$ and $h$, so are smooth. For system \ref{model:part1} at $t \in [0,t_{ub})$, as $k(t)$ is smooth and bounded, the system is smooth if $x$ is finite. It is easy to see that $x$ is finite, System \ref{model:part2} is bounded, and system \ref{model:part1} is bounded for time $t<t_{ub}$. At time $t=t_{ub}$, if $k(t_{ub})$ is finite then the system switches from $(r - h) x \left(1 - \frac{x}{k(t)}\right)$ to $(r - h) x $, so the transition is smooth.
Now, if $k(t_{ub}) \to \infty$ then the transition is continuous. So, the system is continuous. For $t \in [0,t_{ub})$, We have, $\frac{\partial h'}{\partial h}=0$,$\frac{\partial h'}{\partial x}=a$,
$\frac{\partial x'}{\partial h}=-x(1-\frac{x}{k(t)})$ and $\frac{\partial x'}{\partial x}=(r-h)(1-\frac{2x}{k(t)})$ are bounded as $x$ and $h$ are finite at $t \in [0,t_{ub})$. So system \ref{model:part1} is locally Lipschitz. For $t \in [t_{ub},t_{end}]$, system is linear on $x$ and $h$, so system \ref{model:part2} is globally Lipschitz. By Picard-Lindelöf theorem, as the system is continuous and locally Lipschitz, the system has local existence of solutions. 

For \textit{uniqueness of solutions and continuous dependence on initial condition}:
The system \ref{model:part1}-\ref{model:part2} is locally Lipschitz. So, $(h,x)$ are unique in a neighborhood for any initial condition $(h_0,x_0)$. At $t=t_{ub}$, the system is continuous as $k(t)$ is smooth and bounded, so uniqueness is preserved. So, the solutions are unique for $t \in [0,t_{end}]$. As $k(t)$ is smooth in $k_{max},k_{min}$, and the system polynomials are smooth in $a,r$, so a small change in the parameter leads to a small change in the solution. Thus, continuous dependence on initial conditions and parameters is established.
Thus, system, \ref{model:part1}-\ref{model:part2} is well-posed for $t\in [0,t_{end}]$.
\end{proof}

\begin{figure}
    \centering
\includegraphics[width = 6.5cm]{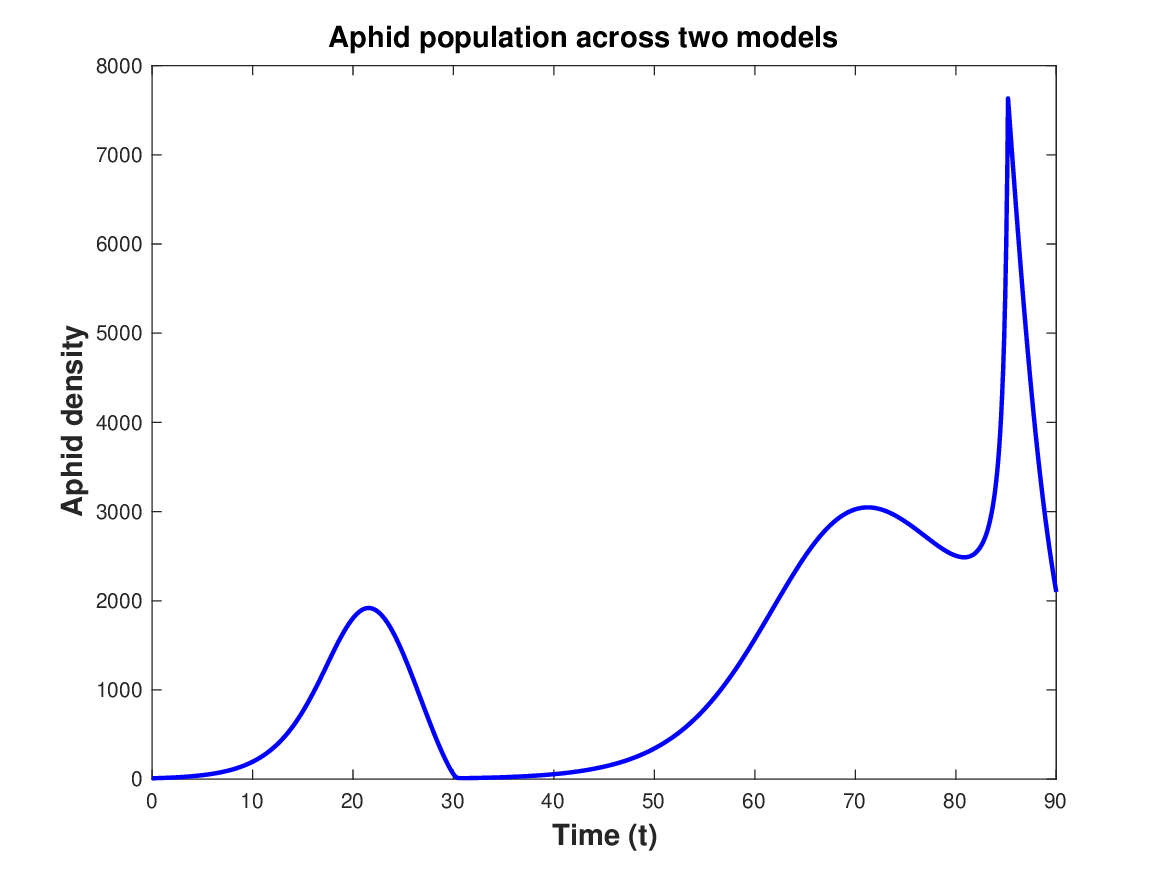}
    \caption{Multiple peaks in the aphid population can be seen for the combined model \eqref{model:part1}-\eqref{model:part2}. The combined model was simulated for a time span of 90 days, where $t_{ub}=85.2$ is when the switch was made. Until this time step $t_{ub}$, the model \eqref{model:part1} was simulated, and then from the next time step until $t_{end}=90$, the model represented by equations \eqref{model:part2} was simulated. The initial populations at $t=0$ were chosen as $h(0)=0,x(0)=10,$ and then the initial condition for model \eqref{model:part2} was taken from the populations at the last evaluated time step so, $h_{ub}=0.5040, x_{ub}=7632.6999$. The parameters used are $K_{max} =10000, K_{min} =1, d =.033, a = 0.000005, r=0.3.$}
    \label{fig:adaptive_model}
\end{figure}

\section{Non-autonomous variable fitness Model}
\label{Non-autonomous Model}
A biologically realistic way of modeling the in-season population dynamics is by building upon the empirical evidence that the environmental factors like flooding and drought affect the host plant quality as seen in Lewis et.al \cite{lewis2025host}. We develop a non-autonomous model incorporating time-dependent plant fitness and pest growth rate. We assume the functions to be periodic, as observed in nature, to be $T$-periodic. We assume the plant fitness as $a(t)$ and the time-dependent growth rate of the pest population as $r(t)$. Thus, we formulate a non-autonomous model as follows :

	\begin{equation}\label{am1}
		\begin{cases}
			\frac{dh(t)}{dt} = a(t)x(t), \\
			\frac{dx(t)}{dt} = \big(r(t) - h(t)\big)x(t)
		\end{cases}
	\end{equation}
	where $a(t)$ and $r(t)$ are positive $T$-periodic functions.
	
	\subsection{Nonexistence of $T$-Periodic Solutions}
	\begin{lemma}
    \label{lem:tp1}
		System \eqref{am1} has no positive $T$-periodic solution, i.e., there does not exist a pair of continuously differentiable functions $(h(t), x(t))$ such that
		\[
		h(t+T) = h(t), \quad x(t+T) = x(t) \quad \text{for all } t \in \mathbb{R}.
		\]
	\end{lemma}
	
	\begin{proof}
		Suppose, for contradiction, that there exists a positive $T$-periodic solution $(h(t), x(t))$ of system \eqref{am1}, with $x(t) > 0$ for all  $t \in \mathbb{R}$, and $h(t) \geq 0$.
		
		From the first equation in \eqref{am1},
		\[
		\frac{dh(t)}{dt} = a(t) x(t) > 0,
		\]
		since both $a(t) > 0$ and $x(t) > 0$. Hence, $h(t)$ is strictly increasing. Therefore, for any $T > 0$,
		\[
		h(t + T) = h(t) + \int_t^{t+T} a(s) x(s)\, ds > h(t).
		\]
		This contradicts the assumption that $h(t)$ is $T$-periodic. Hence, the nonexistence of a periodic solution is proved.
	\end{proof}
	
	\begin{remark}
		This result holds regardless of whether $a(t)$ or $r(t)$ are constant or periodic, as long as $a(t)>0, a>0, x(t)>0$.
		$\frac{dh}{dt}>0$ forces strict growth in $h(t)$, which prevents periodicity. 
	\end{remark}
	
	\subsection{Finite-Time Transient Periodicity and Eventual Extinction}
	\begin{lemma} \label{transient_proof}
		There exists a finite time $T^{*}>0$, such that 
		\begin{enumerate}
			\item[(i)] $t\in [0,t^{*}]$, $x(t)$ exhibits multiple oscillations or local extrema, reflecting the periodicity of $a(t)$, $r(t)$.
			\item[(ii)] For all $t > T^{*}$, we have $h(t) > sup \  r(t)$, and therefore
			\[
			\frac{dx}{dt} = (r(t) - h(t))x(t) < 0, \quad \text{implying } x(t) \to 0 \text{ as } t \to \infty.
			\]
		\end{enumerate}
	\end{lemma}
	
\begin{proof}
		Since $ x(t) > 0 $, and $ a(t) > 0 $, it follows that
		\[
		\frac{dh(t)}{dt} = a(t)x(t) > 0,
		\]
		clearly, $h(t)$  is strictly increasing.
		Hence, there exists a time $T^* > 0$ such that
		\[
		h(t) > \sup r(t) \quad \text{for all } t > T^*.
		\]
		For $t \in [0, T^*]$, the function $r(t) - h(t)$ fluctuates due to the periodic nature of $r(t)$ and the fact that $h(t)$ starts small. Therefore, $x(t)$ experiences alternating growth and decay depending on the sign of $r(t) - h(t)$, giving rise to transient periodicity with multiple peaks. However, after $t > T^{*}$, the inequality $h(t) > r(t)$ ensures
		\[
		\frac{dx}{dt} = (r(t) - h(t))x(t) < 0,
		\]
		leading to exponential decay of $x(t)$ to zero. Thus, the periodic behavior is only maintained transiently on $[0, T^*]$.
	\end{proof}

	\begin{figure}[h]
\begin{subfigure}{.5\textwidth}
\centering
\includegraphics[width= 7cm, height=5cm]{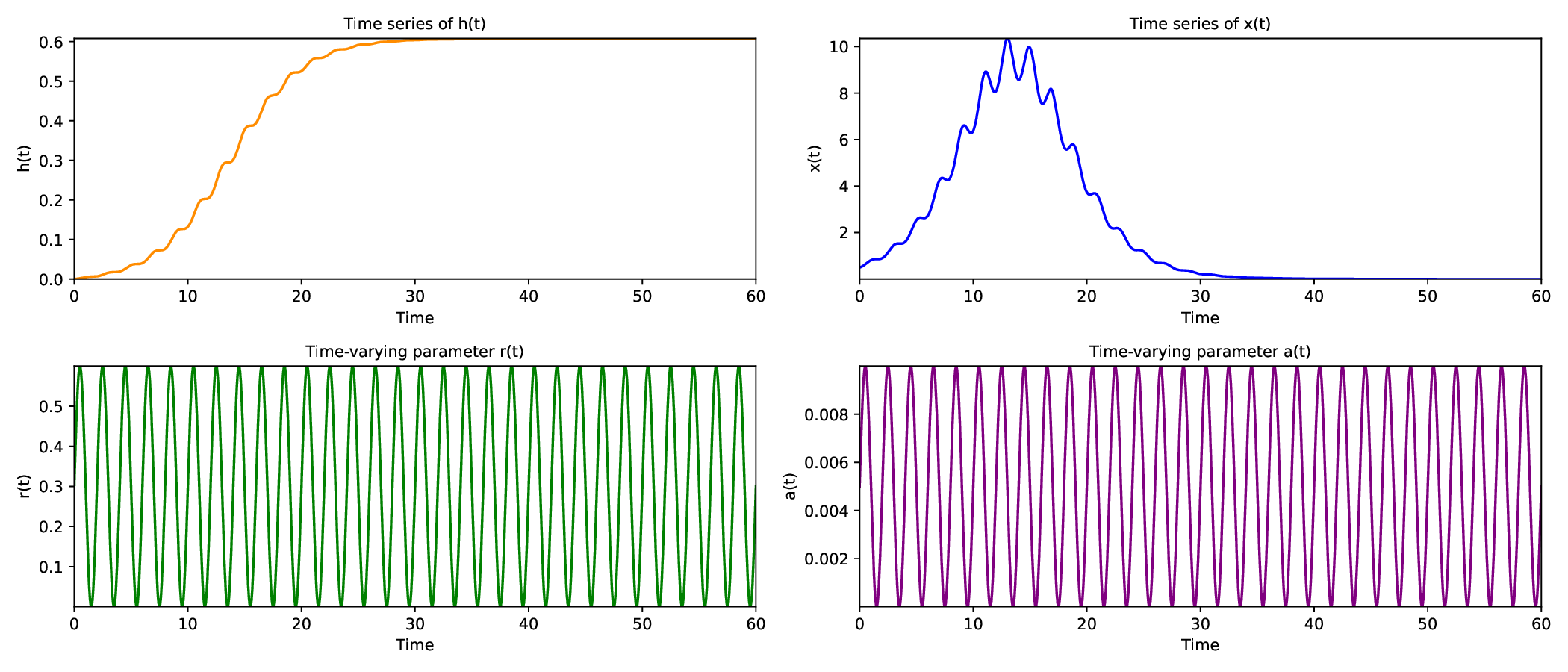 }
\subcaption{}
%need to explain param and fig
%param used K_max = 10000; K_min =1; d =.033;a = 0.000005; r=0.3;x_0=50
\label{f1}
\end{subfigure}
\begin{subfigure}{.5\textwidth}
\centering
\includegraphics[width= 6cm, height=5cm]{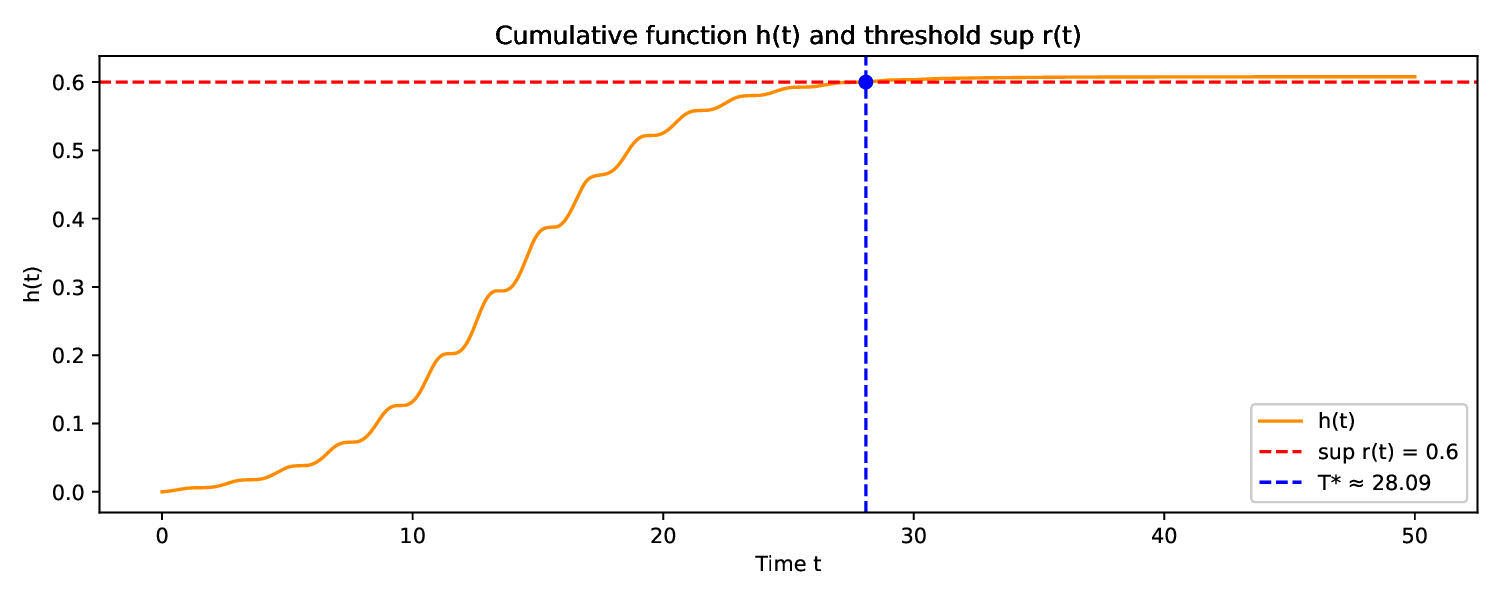 }
\subcaption{}
\label{f2}
\end{subfigure}

\begin{subfigure}{.5\textwidth}
\centering
\includegraphics[width= 6cm, height=5cm]{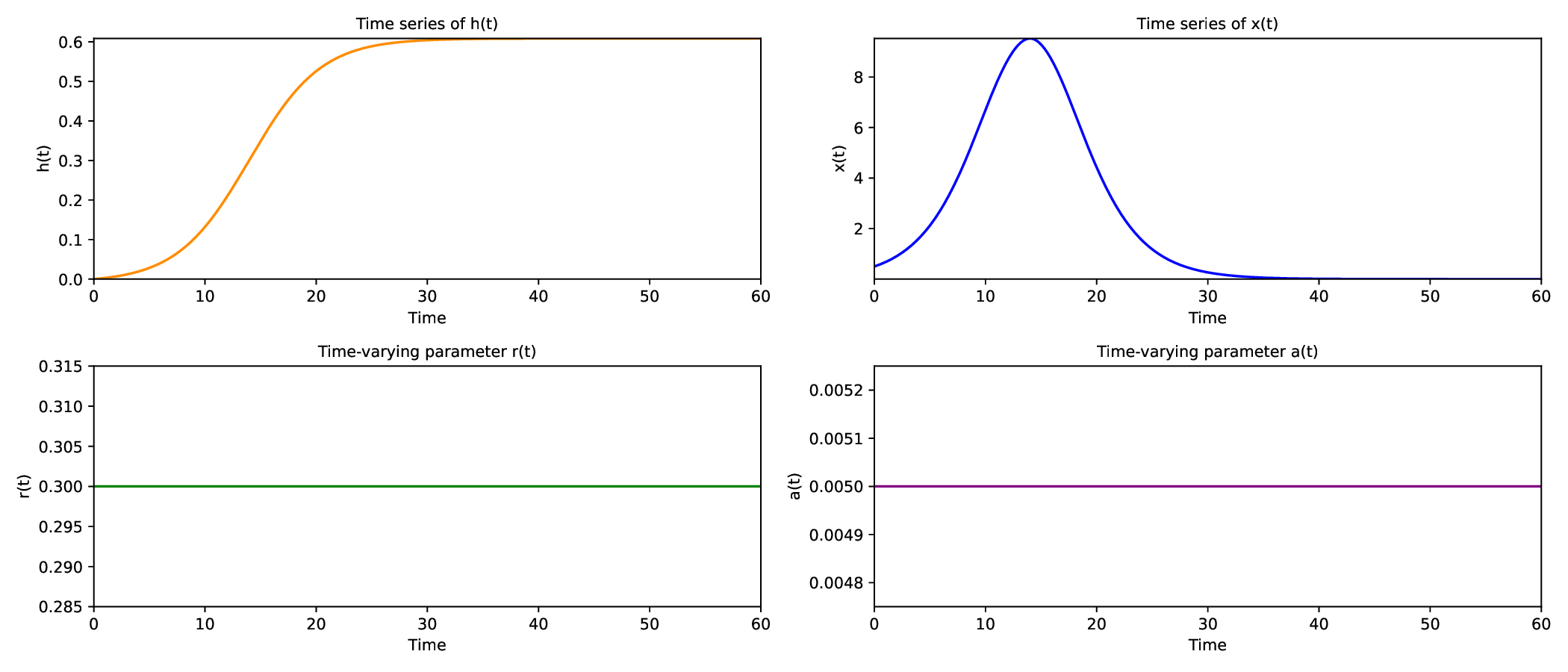 }
\subcaption{}
\label{f3}
\end{subfigure}
\begin{subfigure}{.5\textwidth}
\centering
\includegraphics[width= 6cm, height=5.3cm]{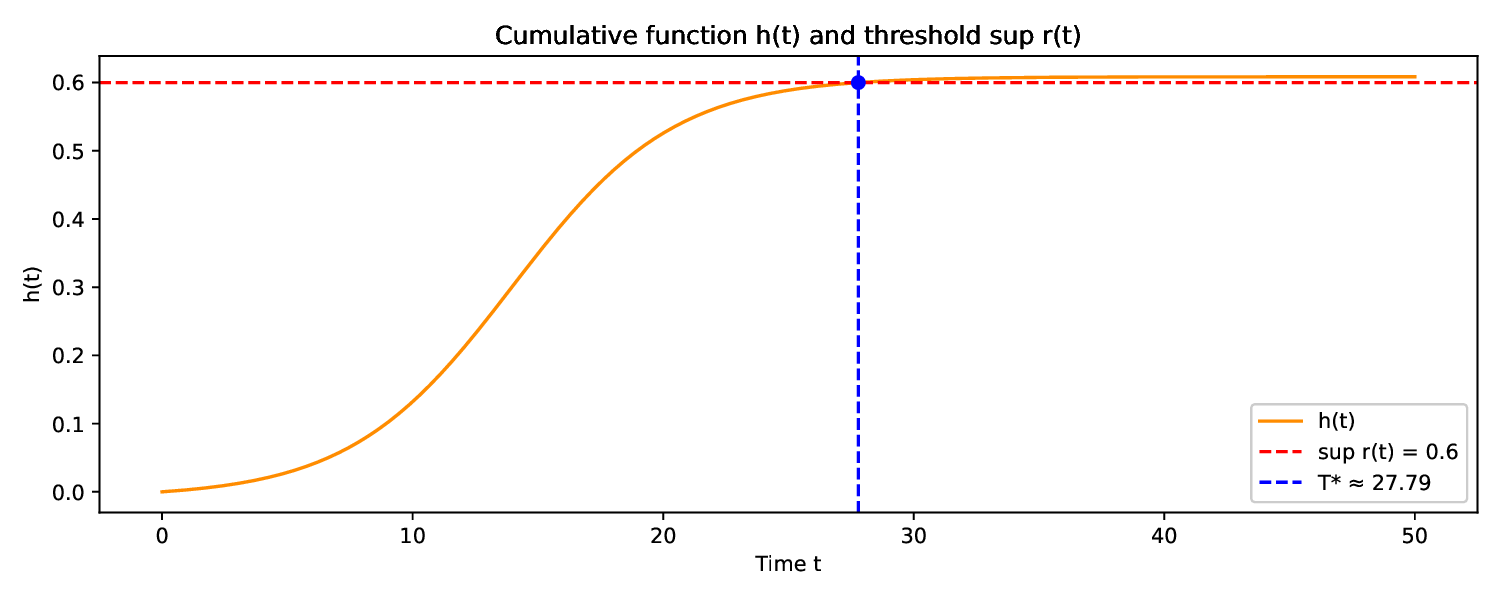 }
\subcaption{}
\label{f4}
\end{subfigure}

\caption{For Figure  \ref{f1}: The time series of the cumulative aphid population $h(t)$, the soybean aphid population $x(t)$ and the time dependent parameters $a(t)$ and $r(t)$ are shown in the figure for model \eqref{am1}. The parameters and initial conditions taken fixed for all the figures are $r_0=0.3, a_0=0.005, \ \omega=q\pi,\ q=0.3,\ x(0)=0.5, \ h(0)=0, r(t)= r_0 (1+\sin{\omega t}), a(t)= a_0 (1+\sin{\omega t})$.  It is seen that the aphid population shows a transient periodic solution with multiple peaks, with both $r(t)$ and $a(t)$ being time-dependent. For Figure \ref{f2}: Time series of the cumulative function $h(t)$ (orange curve) along with the horizontal threshold line $\sup r(t) = 0.6$ (red dashed). The vertical dashed line indicates the critical time $T^* \approx 28.09$, at which $h(t)$ first exceeds  $\sup r(t)$. This marks the end of the transient periodic regime and the beginning of monotonic decay in the population $x(t)$. For Figures \ref{f3}--\ref{f4}: Simulations when both parameters are taken as constants, i.e., $r(t)=r_0$ and $a(t)=a_0$. In this case, the system does not exhibit any transient periodic dynamics.}
\label{am_figure}
\end{figure}

\begin{remark}
Cumulative variable $h(t)$ increases monotonically and never decreases, confirming it cannot be periodic. The population $x(t)$ shows transient oscillations (multi-peaked pattern), but each peak gets smaller, showing decay over time.
	\end{remark}
	
	\begin{remark}
		For the specific case where $a(t) = 0.005(1 + \sin(\pi t))$, $r(t) = 0.3(1 + \sin(\pi t))$, with initial conditions $h(0) = 0$, $x(0) = 0.5$, numerical simulation (see Figure \ref{f2}) shows that the cumulative variable $h(t)$ exceeds the maximum of $r(t)$ (i.e., $\sup r(t) = 0.6$) at approximately $T^* \approx 28.09.$ This confirms that the solution $x(t)$ exhibits transient periodic dynamics on the interval $[0, T^*]$, after which it decays monotonically to zero due to the dominance of the cumulative term $h(t)$.
        	\end{remark}

\section{Economic Threshold (ET) and Economic Injury Level (EIL)}
\label{ET_EIL}
In this section, we present various simulations and a table summarizing results for models \eqref{eq:cl1}, \eqref{model:part1}-\eqref{model:part2}, and
 \eqref{am1}. The common parameters (the scaling parameter $a$ and the growth rate of aphids $r$) were kept the same, and we noted when soybean aphids surpass the ET ($250$ aphids per plant)  and EIL level ($674$ aphids per plant) \cite{ragsdale2011ecology} for all three models.  

 For Figure~\ref{et_eil_2}, parameters were chosen so that the average aphid population in model~\eqref{eq:cl1} fell below the ET. The same parameter set was then used to evaluate the dynamics and thresholds for models \eqref{model:part1}–\eqref{model:part2} and \eqref{am1}. In both Figures~\ref{et_eil} and \ref{et_eil_2}, models \eqref{model:part1}–\eqref{model:part2} and \eqref{am1} show multiple peaks in aphid density.

Table \ref{table:et_eil_levels}-\ref{table:et_eil_levels_2} summarize these results for figures \ref{et_eil} and \ref{et_eil_2} respectively.  We also report the peak aphid population and average aphid population over a single season for all models. It can be seen that \eqref{model:part1}-\eqref{model:part2} achieved the lowest peak aphid population; however, this did not translate into a low seasonal average because multiple peaks occurred at high densities. Whereas model \eqref{eq:cl1} and \eqref{am1} exhibited higher peaks, their seasonal averages remained within the same range as each other. 
\begin{table}
 \centering
\begin{tabular}{|l|c|c|c|c|}
\hline
\textbf{Model} & \textbf{Peak Aphid Population } & \textbf{ET exceeded on} & \textbf{EIL exceeded on} & \textbf{Average Population} \\
\hline
\eqref{eq:cl1}  & 9004.2221 (on Day 27)  & $\approx$ Day 11 & $\approx$ Day 15 & 1332.014225\\
 \hline
\eqref{model:part1}-\eqref{model:part2}& 7632.6999 (on Day 85)  & $\approx$ Day 11 &  $\approx$ Day 14 & 1358.530699 \\
 \hline
\eqref{am1} & 10150.4484 (on Day 23)  & $\approx$ Day 9 &  $\approx$ Day 14 & 1316.262857
%Average aphid population over the combined time span: 1358.530699
%Time when population first exceeds 250 in combined model: 10.8686
%Time when population first exceeds 674 in combined model: 14.5176
%Peak population in combined model: 7632.6999 at time 85.2000
 \\
\hline
\end{tabular}
\caption{ The table shows peak population, the day when Economic Threshold (ET) and Economic Injury Level (EIL) were crossed, and average aphid population over a single season for Figure \ref{et_eil}. The peak aphid population for model \eqref{model:part1}-\eqref{model:part2} is calculated by combining the aphid population for the complete time span across two models. }
%\refstepcounter{table}
\label{table:et_eil_levels}
\end{table}
 \begin{figure}
\begin{subfigure}{.35\textwidth}
  \centering
  \includegraphics[width= 6.3cm, height=5cm]{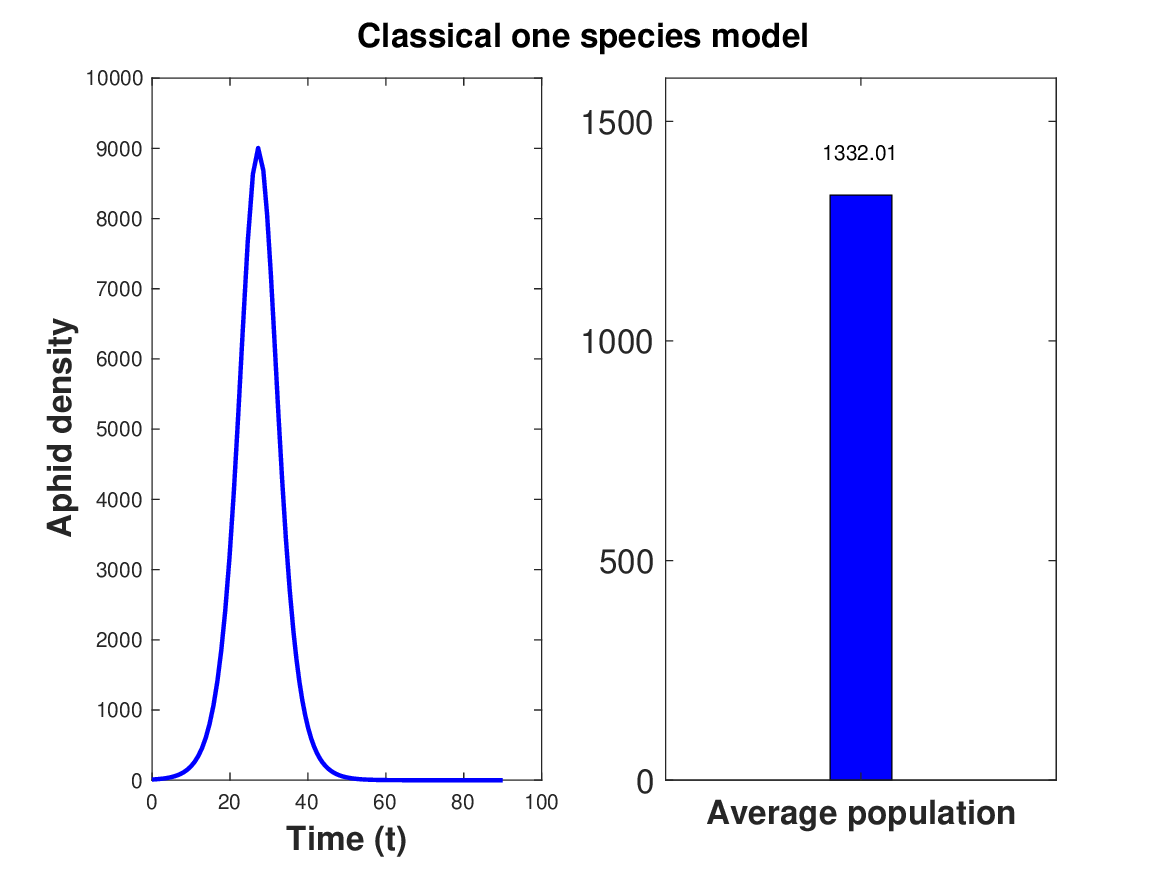}
  \caption{}
  \label{fig:et_fig1}
\end{subfigure}%
\begin{subfigure}{.35\textwidth}
  \centering
  \includegraphics[width= 6.3cm, height=5cm]{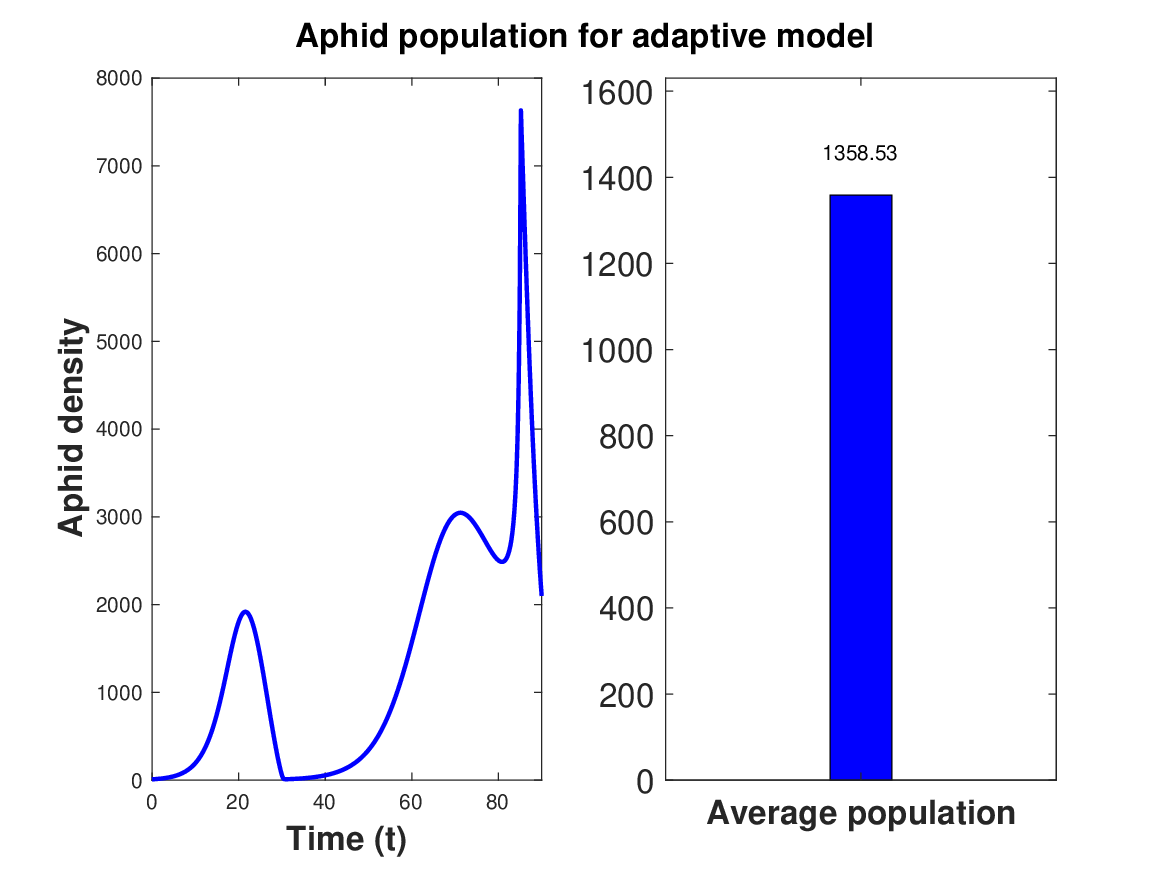}
  \caption{}
  \label{fig:et_fig2}
\end{subfigure}%
\begin{subfigure}{.35\textwidth}
  \centering
  \includegraphics[width= 6.3cm, height=5cm]{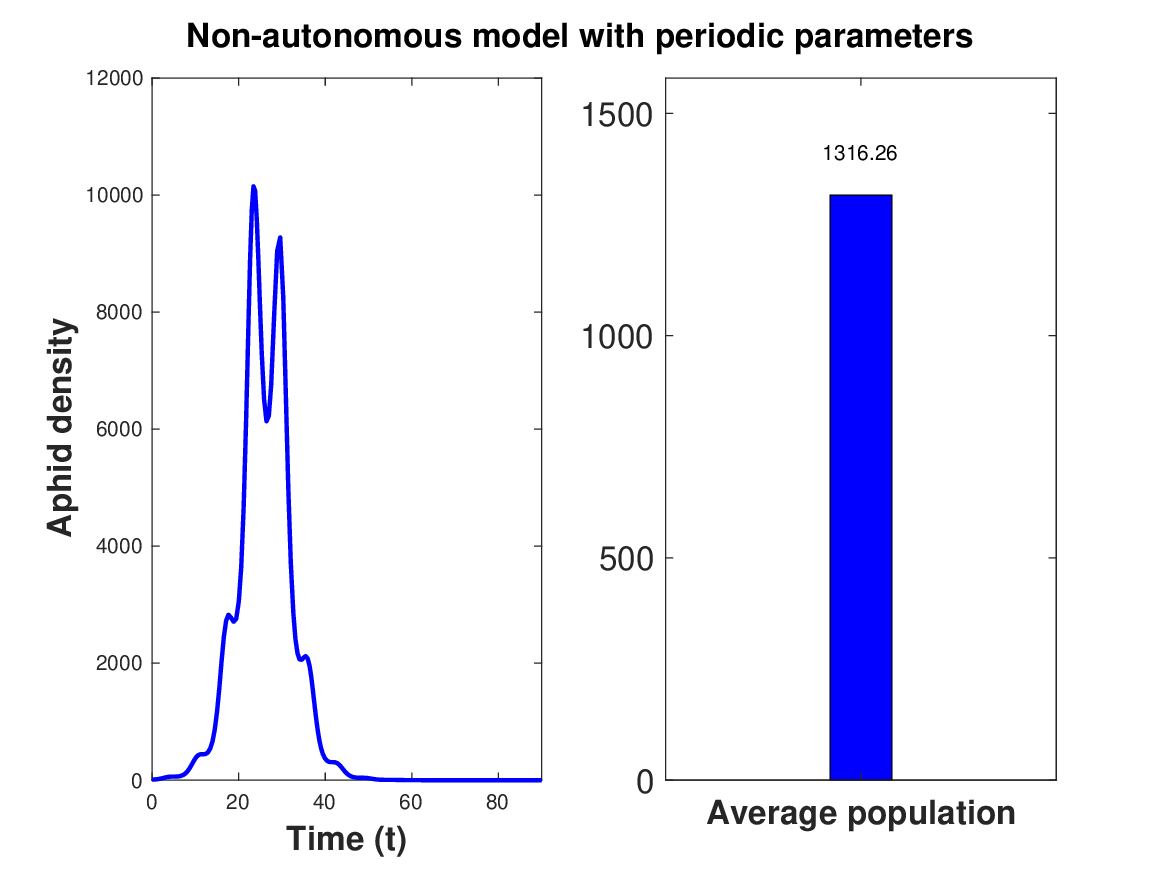}
  \caption{}
  \label{fig:et_fig3}
\end{subfigure}%
\caption{The common parameters used for all three models are $r=r_0=0.3, a=a_0=0.000005$. The initial population and time span are also kept the same with $h(0)=0,x(0)=10$ and $t=[0,90]$. The parameters for the adaptive model \eqref{model:part1}-\eqref{model:part2} are $K_{max} = 10000, K_{min} =1, d =.033$ and for non-autonomous model with periodic parameters \eqref{am1}, $r(t)= r_0 (1+\sin{\omega t}), a(t)= a_0 (1+\sin{\omega t}), q=0.3, \omega=q\pi$. In Figure \ref{fig:et_fig1}, it can be seen that aphid populations exhibit boom-bust dynamics with one peak, whereas for Figures \ref{fig:et_fig2} and \ref{fig:et_fig3}, multiple peaks can be seen throughout the time span. The peak aphid population for Figure \ref{fig:et_fig1} and Figure \ref{fig:et_fig3} occurs above $9000$; however, for Figure \ref{fig:et_fig2}, the peak is achieved around $7600$ and towards the end of the season. The Economic Injury level (EIL) was crossed around day $14$ for all three figures.} 
\label{et_eil}
\end{figure}
\begin{table}
 \centering
\begin{tabular}{|l|c|c|c|c|}
\hline
\textbf{Model} & \textbf{Peak Aphid Population } & \textbf{ET exceeded on} & \textbf{EIL exceeded on} & \textbf{Average Population} \\
\hline
\eqref{eq:cl1}  & 2034.1058 (on Day 15)  & $\approx$ Day 7 &   $\approx$ Day 10 & 199.803415\\
%Time when population first exceeds 250(ET): 7.4751
%Time when population first exceeds 674(EIL): 10.1907
%Peak population: 2034.1058
%Peak population: 2034.1058 at time 14.9816

 \hline
\eqref{model:part1}-\eqref{model:part2} & 1741.6799 (on Day 15)  & $\approx$ Day 7  & $\approx$ Day 10 & 244.001675 \\
%eil on day 10th
% first model run until t= 29.2 and then next time step kindlmann classical, the initial condition in the beginning was 0,10, params are a=0..00005, r=0.45, ic for next model was 0.9883;1297.8316, rest all parameters are same K_max =10000;K_min =1;d =.033; 
 \hline
\eqref{am1} & 2586.5328 (on Day 16)  & $\approx$ Day 7  & $\approx$ Day 8 & 200.082577 
%eil on Day 8.6444th
 \\
\hline
\end{tabular}
\caption{ The table shows peak population, the day Economic Threshold (ET) and Economic Injury Level (EIL) were crossed, and average aphid population over a single season for Figure \ref{et_eil_2}. The peak aphid population for model \eqref{model:part1}-\eqref{model:part2} is calculated by combining the aphid population for the complete time span across two models. }
%\refstepcounter{table}
\label{table:et_eil_levels_2}
\end{table}
\begin{figure}
\begin{subfigure}{.35\textwidth}
  \centering
  \includegraphics[width= 6.3cm, height=5cm]{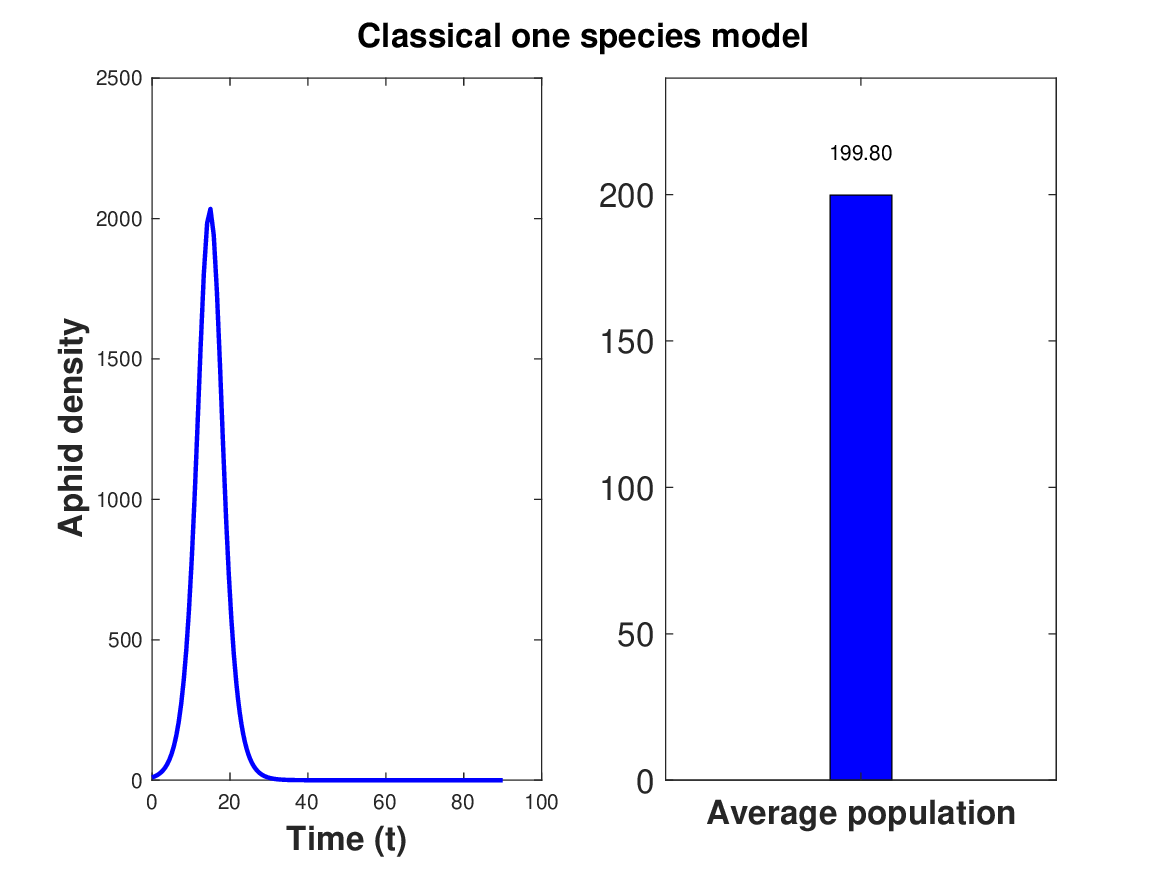}
  \caption{}
  \label{fig:et_fig_1}
\end{subfigure}%
\begin{subfigure}{.35\textwidth}
  \centering
  \includegraphics[width= 6.3cm, height=5cm]{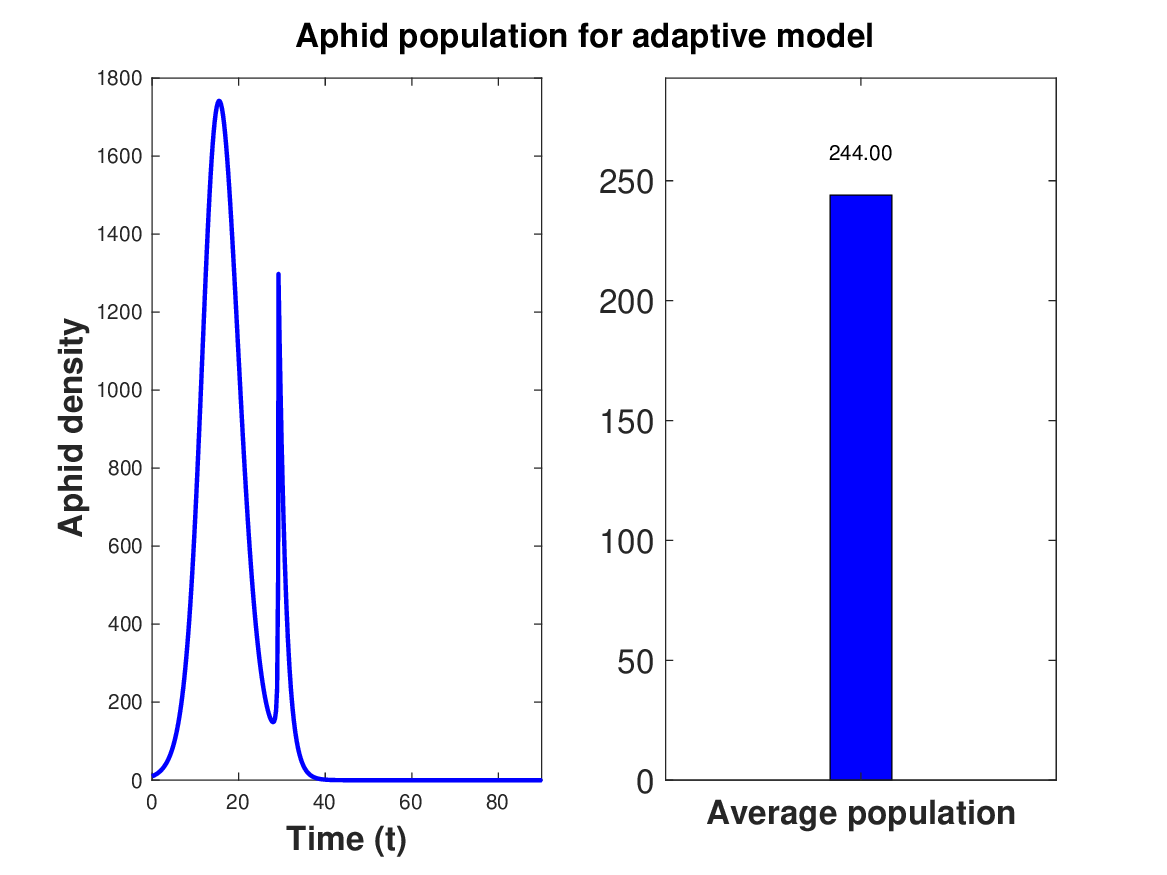}
  \caption{}
  \label{fig:et_fig_2}
\end{subfigure}%
\begin{subfigure}{.35\textwidth}
  \centering
  \includegraphics[width= 6.3cm, height=5cm]{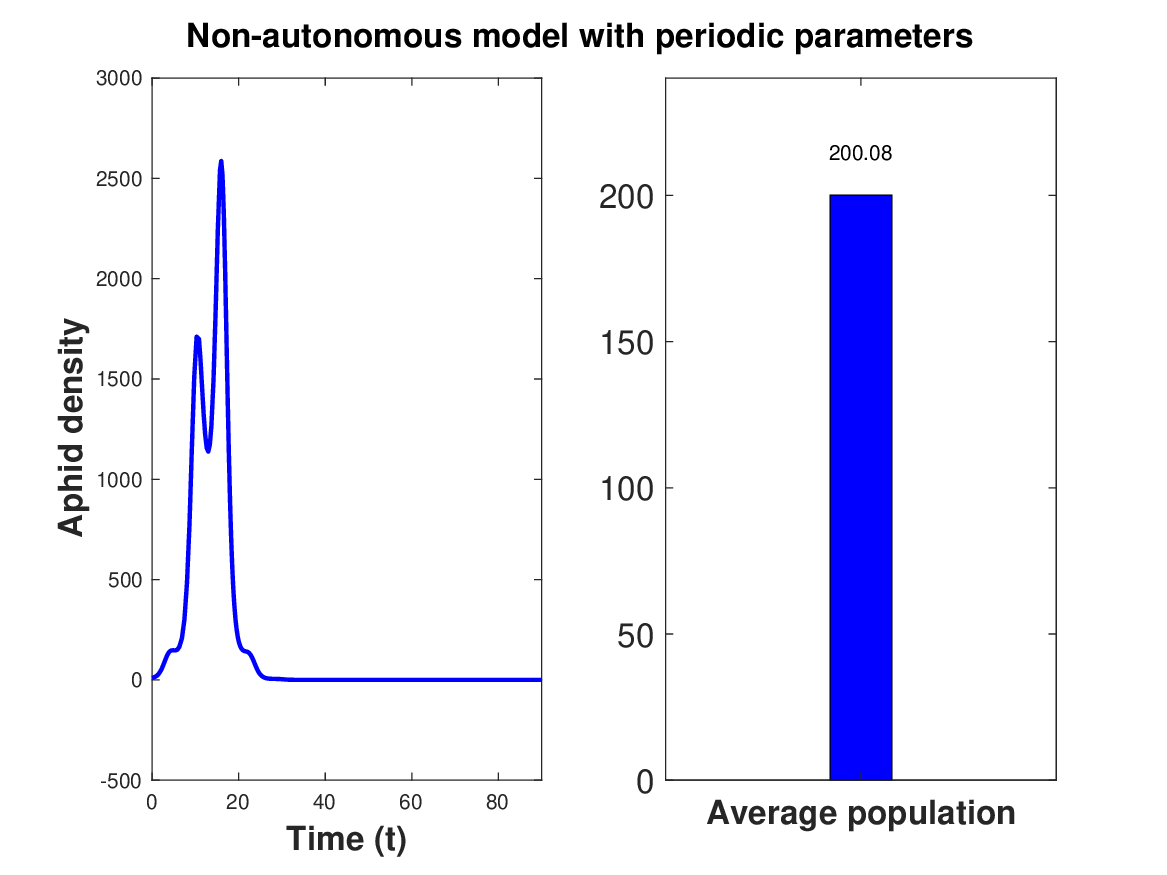}
  \caption{}
  \label{fig:et_fig_3}
\end{subfigure}%
\caption{The common parameters used for all three models are $r=r_0=0.45, a=a_0=0.00005$. The initial population and time span are also kept the same with $h(0)=0,x(0)=10$ and $t=[0,90]$. The parameters for the adaptive model \eqref{model:part1}-\eqref{model:part2} are $K_{max} = 10000, K_{min} =1, d =.033$ and for non-autonomous model with periodic parameters \eqref{am1}, $r(t)= r_0 (1+\sin{\omega t}), a(t)= a_0 (1+\sin{\omega t}), q=0.3, \omega=q\pi$. In Figure \ref{fig:et_fig_1}, it can be seen that aphid populations exhibit boom-bust dynamics with one peak, whereas for Figures \ref{fig:et_fig_2} and \ref{fig:et_fig_3}, multiple peaks can be seen throughout the time span. Due to larger values of $a$ and $r$, peaks in \ref{fig:et_fig_2} appear earlier in the season when compared to Figure \ref{fig:et_fig2}. 
%The peak aphid population for Figure \ref{fig:et_fig_1} and Figure \ref{fig:et_fig_3} occurs above $9000$; however, for Figure \ref{fig:et_fig_2}, the peak is achieved around $3000$ and much later in the season. 
The Economic Threshold was crossed around day $7$ for all three figures.} 
\label{et_eil_2}
\end{figure}

\section{Discussion And Conclusion}
\label{Discussion And Conclusion}
The classical model for aphid dynamics proposed by Kindlmann et. al. \cite{kindlmann2010modelling} accurately predicts the ``boom-bust" single peak dynamics - such as seen in soybean aphid pest populations. However, population dynamics of aphids on several other crops display multiple peaks (two to three) in a single growing season - a dynamic which cannot be captured by classical models \cite{kindlmann2010modelling}. This paper explores several real-world motivations for multiple peak dynamics, and offers alternative modeling methods to address this issue. The economic damage to crop yield, and the specific levels at which such damage occurs due to aphid infestation, has also been examined. 

The current work reaffirms that the VCM proposed in \cite{K07} has extremely rich dynamics, and is capable of predicting multiple population peaks, as is observed in the field with many species of aphids (Fig. \ref{fig:klm_x0_80}). However, the model also possesses finite time blow-up solutions (Fig. \ref{fig:klm_x0_85}) that have not been reported earlier in the literature, to the best of our knowledge. Essentially we show that the VCM \cite{K07}, has two critical features: it is able to correctly reproduce multi-peak dynamics of aphid populations (Fig. \ref{fig:klm_x0_80}), but it is also subject to finite-time blow-up, a new finding - not previously reported for this model. Our analysis of the blow-up behavior, through a numerical study of \eqref{eq:1} - \eqref{eq:22a}, confirms that its initiation is catalyzed by higher initial populations ($x_0$), and is triggered by fluctuations in the $a$ parameter. These results demonstrate a fundamental structural instability that remains even when more realistic factors, e.g., time delay, are included in the model structure. Furthermore, from a purely mathematical angle, one can investigate deriving the necessary conditions for finite-time blow-up solutions in the VCM. We see numerically that the smaller the parameter $a$ is, the smaller the data required to produce a finite-time blow-up solution is. This is unproven by us presently, and warrants further investigation.

Due to the blow-up dynamic intrinsic to the VCM, in section \ref{Adaptive behavior model} we have formulated a novel adaptive behavioral model (\ref{model:part1}) - (\ref{model:part2}), that possesses bounded solutions, for all initial data and parameter regimes. This model includes a mechanism whereby the ``surge" in population initiates a behavior adaptation, efficiently eliminating  blow-up and stabilizing growth, based on a classical model. Herein we are motivated by earlier ideas in the literature, for similar  agro-pest systems \cite{ainseba2023}. Furthermore, we demonstrate that this switch-like behavior leads to a well-posed model. As shown by numerical simulations (Fig. \ref{fig:adaptive_model}), our new method still preserves the capacity to produce multiple population peaks, providing a stable framework for further analysis of aphid dynamics.

Building on recent empirical research by Lewis et al. \cite{lewis2025host}, we formulate a non-autonomous model (\ref{am1}) that includes time-dependent plant fitness and pest growth rates in section \ref{Non-autonomous Model}. This structure accommodates the dynamic coupling between environmental conditions—like flooding demonstrated to influence soybean aphid biotypes—and population expansion. Analytic confirmation as in figure \ref{am_figure} demonstrates that through T-periodic fitness functions, the model reproduces transient multi-peak dynamics reported in the field. Although the analysis of non-autonomous systems is by necessity complicated, the analysis in Lemma \ref{transient_proof} provides an important step in correlating environmental change with population outcomes directly, hence creating a significant framework for predicting climate change responses in pests. Note, although Lemma \ref{lem:tp1} shows that there do not exist T-periodic solutions to (\ref{am1}), future work could investigate the existence of such solutions via other factors, such as structure or control agents. Herein the mathematical techniques of coincidence degree theory would be the tools of choice \cite{mawhin2006topological}. This has been undertaken in a number of recent ecological problems \cite{srivastava2023periodic, srivastava2025existence}.

The models in all of the sections have been compared in section \ref{ET_EIL} to understand the economic injury level and to analyze the economic threshold above which economic losses from harvesting are predicted. Across all models, the economic threshold (ET) and economic injury level (EIL) offer helpful insights for understanding simulated aphid dynamics. Based on the scaling parameter and growth rate, the population peak can occur either earlier or later in the season for \eqref{model:part1}-\eqref{model:part2}, but this combined model exhibits a higher seasonal average compared to models  \eqref{eq:cl1} and \eqref{am1}, see Tables \ref{table:et_eil_levels}-\ref{table:et_eil_levels_2}. However, models \eqref{eq:cl1} and \eqref{am1} displayed similar timing for peak population occurrence and comparable average densities. In all the simulations of section \ref{ET_EIL}, see figures \ref{et_eil}-\ref{et_eil_2}, aphids ultimately surpass the critical thresholds within a season, suggesting that control strategies are required to keep aphid populations below economically damaging levels. These can be done in many ways, such as the use of foliar insecticide applications \cite{ragsdale2011ecology}, biological control strategies, \cite{costamagna2006predators,lin2003effect,miksanek2019matrix}, or the use of resistant soybean varieties, essentially by moving towards the IPM (Integrated pest management) approach \cite{verma2025towards}. Since the Economic Threshold (ET) shows when control action is needed,   regular scouting is important to avoid excessive or unnecessary use of insecticides or biological control agents.  Incorporating these strategies into models will strengthen their utility for guiding management decisions.

 In the current work we have explored various alternative population models for aphid dynamics, which take into account variable fitness of the host plant due to abiotic factors. These include (i) varying host plant suitability, (ii) weather events such as floods or droughts. In particular, events that happen early in the season and possibly dictate the fitness of the plant to some extent during the season have been explored. A possible future direction is inclusion of resistance mechanisms, such as via use of soybean varieties expressing RAG (Resistance to Aphis Glycines) type genes. These have caused the soybean aphid to evolve into subpopulations of virulent (those that can survive on RAG plants) versus avirulent (those that cannot). Structuring the aphid populations into virulent and avirulent particularly to test our adaptive model on RAG plants would be a valuable future direction. It would be interesting to investigate if only the virulent subpopulation possessed the blow-up dynamic - while the avirulent component featured the adaptive mechanisms - as this would be more realistic on a RAG plant. Further directions, drawing from the results of section \ref{ET_EIL}, could be the inclusion of multiple agents of control, such as predators, parasitoids as well as movement mechanisms such as drift and dispersal, as well as non-linear stocking, harvesting and competition \cite{verma2023t, verma2025towards, banerjee2025two, parshad2021some}. These often provide counterintuitive dynamics, which could be ultimately beneficial for devising pest management tactics and strategies. 
All in all, an accurate construction of biologically sound population models, considering all of the nuances involved in the relevant Aphid biology, remains our primary objective for current and future work.

%\textcolor{red}{Tie everything up for work impact, which shows different modeling techniques for modeling real-time pest dynamics. Then show some future directions. Maybe used to motivate some work for the pest project (Rana).}

\section{Appendix}
We discuss some preliminary well-established results that ensure the non-negativity of solutions and establish both local and global existence, as outlined in \cite{MP10, Hen84},

\begin{Lem}
\label{lem:class1}
	Let us consider the following $m\times m$ - ODE system: for all $i=1,...,m,$ 
	\begin{equation}
		\label{eq:class1}
		\frac{ d X_i}{dt}=f_i(X_1,...,X_m)~in~ \mathbb{R}^{n} \times \mathbb{R}_+ , \ X_i(0)=X_{i0},
	\end{equation}
	where $f=(f_1,...,f_m):\mathbb{R}^m \rightarrow \mathbb{R}^m$ is $C^1(0,T)$ and $X_{i0}\in L^{\infty}(0,T)$, $\forall T$. Then there exists a $T>0$ and a unique classical solution of (\ref{eq:class1}) on $[0,T).$ If $T^*$ denotes the greatest of these $T's$, then 
	\begin{equation*}
		\Bigg[\sup_{t \in [0,T^*),1\leq i\leq m} ||X_i(t)||_{L^{\infty}(0,T)} < +\infty \Bigg] \implies [T^*=+\infty].
	\end{equation*}
 \end{Lem}

 \begin{Lem}
 \label{lem:qp}
	If the non-linearity in \eqref{eq:class1} $(f_i)_{1\leq i\leq m}$ is  quasi-positive, that is, 
	$$\forall i=1,..., m,~~\forall X_1,..., X_m \geq 0,~~f_i(X_1,...,X_{i-1}, 0, X_{i+1}, ..., X_m)\geq 0,$$
	then $$[\forall i=1,..., m, X_{i0}\geq 0]\implies [\forall i=1,...,m,~ \forall t\in [0,T^*), X_i(t)\geq 0].$$
\end{Lem}

\begin{Lem}\label{lem:class2}
	Using the same notations and hypotheses as in Lemma \ref{lem:class1}, suppose moreover that $f$ has at most polynomial growth and that there exists $\mathbf{b}\in \mathbb{R}^m$ and a lower triangular invertible matrix $P$ with non negative entries such that  $$\forall r \in [0,+\infty)^m,~~~Pf(r)\leq \Bigg[1+ \sum_{i=1}^{m} r_i \Bigg]\mathbf{b}.$$
	Then, for $X_0 \in L^{\infty}(\mathbb{R}_+^m),$ the system (\ref{eq:class1}) has a strong global solution.
\end{Lem}

Under these assumptions, the following local existence result is well known, see \cite{Hen84}.

\begin{Thm}
	\label{thm:class3}
	The system (\ref{eq:class1}) admits a unique, classical solution $(X_{1},X_{2},..X_{m})$ on $%
	[0,T_{\max }]$. If $T_{\max }<\infty $ then 
	\begin{equation}
		\underset{t\nearrow T_{\max }}{\lim }\Big\{ \left\Vert X_{1}(t)\right\Vert
		_{\infty }+\left\Vert X_{2}(t)\right\Vert _{\infty }...+ \left\Vert X_{m}(t)\right\Vert _{\infty }\Big\} =\infty ,  
	\end{equation}%
	where $T_{\max }$ denotes the eventual blow-up time in $\mathbb{L}^{\infty }[0,T_{\max }).$
\end{Thm}

\subsection{Proof of Theorem \ref{thm:t1t}}
\begin{proof}
\label{proof_thm:t1t}
Notice, 
\begin{equation}
\label{eq:3}
h(t) = \int^{t}_{0}a x(s)ds = a \int^{t}_{0} x(s)ds \geq a x(t) => \frac{1}{a}h(t) \geq x(t).
\end{equation}
We make the following lower estimate,
\begin{eqnarray}
&& \frac{dx}{dt} \nonumber \\
&=& (r-h)x\left(1-\frac{x}{k}\right)  \nonumber \\
&=& r x - \frac{r}{k}x^{2} - h x + \frac{1}{k}hx^{2} \nonumber \\
&\geq& r x - \frac{r}{k_{min}}x^{2} - h x + \frac{1}{k_{max}}hx^{2} \nonumber \\
&\geq&  r x - \frac{r}{ k_{min} a} h x - h x + \frac{1}{k_{max}}hx^{2} \nonumber \\
&=& r x -\left( \frac{r}{a k_{min}}  + 1\right) h x + \frac{1}{k_{max}}hx^{2} \nonumber \\
&=& r x -\left( \frac{r}{a k_{min}}  + 1\right) 
\left( \int^{t}_{0}a x(s)ds \right) x + \frac{1}{k_{max}}hx^{2} \nonumber \\
\end{eqnarray}
Now, we proceed by contradiction. Assume $x(t)$ remains bounded on any time interval $[0,T], \ T < \infty$, then via the embedding, $L^{\infty}(0,T) \hookrightarrow L^{1}(0,T)$, we must have that,

\begin{equation}
\label{eq:3e}
 \int^{T}_{0}a x(s)ds \leq T C ||x(t)||_{\infty}.
\end{equation}
Thus, inserting this in the above inequality, which must hold for any $t \in [0, T]$, we have,
\begin{eqnarray}
&& \frac{dx}{dt} \nonumber \\
&\geq& r x -\left( \frac{r}{a k_{min}}  + 1\right) 
\left( \int^{t}_{0}a x(s)ds \right) x + \frac{1}{k_{max}}hx^{2} \nonumber \\
&\geq& r x -\left( \frac{r}{a k_{min}}  + 1\right) 
\left(  T C ||x(t)||_{\infty} \right) x + \frac{1}{k_{max}}ax^{3} \nonumber \\
&\geq& (r-M)x + \frac{1}{k_{max}}ax^{3}
\end{eqnarray}

Where $T C |||x(t)||_{\infty} < M$. However, in this case $x$ blows up in comparison with the ODE $\boxed{y^{'} = C_{3}y^{3} + C_{4}y}$,  with $C_{4} > 0$ or $C_{4} < 0$. This contradicts $x$ being bounded at any $T < \infty$. Thus $x$ must blow up at a finite time $T^{*} \in [0,T]$.
\end{proof}
\subsection{Proof of Theorem \ref{thm:t1}}
\begin{proof}

\label{proof_thm:t1}
Using the lower estimate from theorem \ref{thm:t1t},

\begin{equation}
\frac{dx}{dt} 
\geq r x -\left( \frac{r}{a k_{min}}  + 1\right) h x + \frac{1}{k_{max}}hx^{2} 
\geq \left(  \left( \frac{x}{k_{max}}\right)  - \left( \frac{r  }{a k_{min}}  + 1\right)  \right) h x
\end{equation}

Thus $x$ blows up trivially for sufficiently large data chosen as per the requirement of Lemma \ref{lem:l1}, in comparison with the ODE, $\boxed{y^{'} = \epsilon y^{2}}$,
$\epsilon > 0$. 
\end{proof}
\subsection{Proof of Lemma \ref{Lem:l12}}
\begin{proof}
\label{proof_Lem:l12}
We know from the mean value theorem of integrals,

\begin{equation}
\frac{1}{T} \int^{T}_{0}x(s)ds = x(T^{**}), \ T < \infty, \ T^{**} \in [0,T].
\end{equation}

Via Theorem \ref{thm:t1}, we have the blow-up of $x(t)$ at some $T^{*} < \infty$. Consider,

\begin{equation}
\frac{1}{2(T^{*} - \Delta t)}x(T^{**}) \leq \frac{1}{(T^{*} - \Delta t)} \int^{T^{*} - \Delta t}_{0}x(s)ds = x(T^{**}),
\end{equation}

Now taking the limit as $\Delta t \rightarrow 0$ entails, $T^{**} \rightarrow T^{*}$, thus we have

\begin{equation}
\frac{1}{2(T^{*})}x(T^{*}) = \infty \leq \frac{1}{T^{*}} \int^{T^{*} }_{0}x(s)ds = x(T^{*}) = \infty,
\end{equation}

The result follows via the squeezing theorem.
\end{proof}

\begin{figure}%simulated in Matlab code
\centering
\includegraphics[width = 7cm]{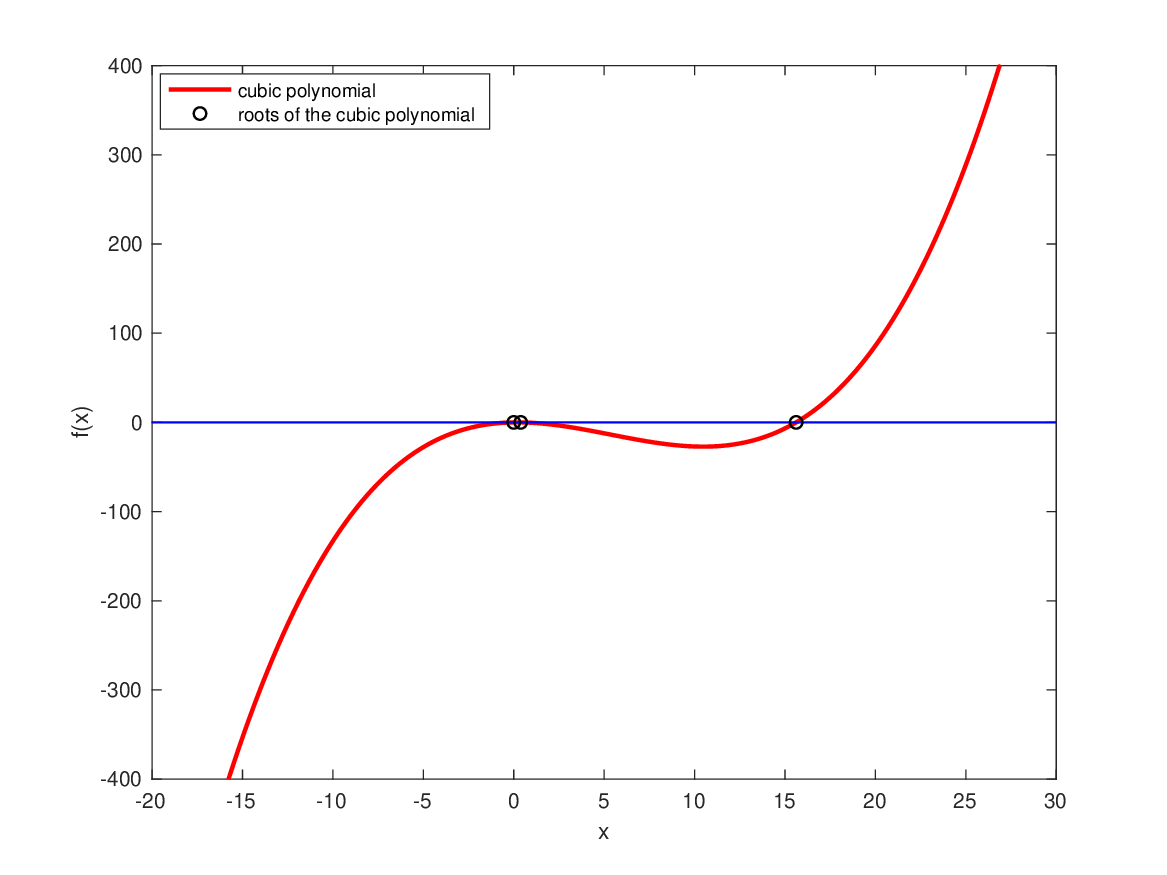 }
\caption{This figure represents the cubic polynomial in (\ref{eq:c3}). The parameter set used: $ K_{max} = 10, K_{min} =1, a = 0.5, r=0.3$. }
\label{plot_poly}
\end{figure}
\subsection{Proof of Theorem \ref{thm:l1k}}
\begin{proof}
\label{proof_thm:l1k}
Assume that there is a maximal time of existence $T^{*}$ to the solution of \eqref{eq:d1}. Via the equation, $y^{'} > 0$, thus $y$ is increasing. Let \begin{equation}
\label{eq:s21}
    t_{1} = \sup \{ t \geq 0, y^{'}(s) > 0, s \in [0,t] \}.
\end{equation} 

If $t_{1} < T^{*}$, then $y^{'}(t_{1}) = 0$. Now, let us consider 2 cases.

\textbf{Case I: $t_{1} < \tau$,} \ Since $y^{'}(t_{1}) = 0$, we have 
$|y(t_{1})|^{p} - M|\phi(t_{1}-\tau)|^{q} = 0.$

This is not possible due to the increasing dynamics of $y$, 
$|y(t_{1})|^{p} > |y(0)|^{p} = |\phi(0)|^{p}  \geq M|\phi(t_{1}-\tau)|^{q}.$

\textbf{Case II: $t_{1} > \tau$,} \ In this case, again since $y^{'}(t_{1}) = 0$, we have 
\begin{equation}
|y(t_{1})|^{p} - M|y(t_{1}-\tau)|^{q} = 0.
\end{equation}

Now $y(t) > y(0) = \phi(0) \geq M \geq 1$, \ So, $
|y(t_{1})| = M^{\frac{1}{p}}|y(t_{1}-\tau)|^{\frac{q}{p}} < M^{\frac{1}{p}}|y(t_{1}-\tau)|. 
$
\

On the other hand, \ $
|y(t_{1})|^{p} > |y\left(t_{1}-\frac{\tau}{2}\right)|^{p}. 
$

Thus, we have, 

\begin{equation}
|y\left(t_{1}-\frac{\tau}{2}\right)|^{p} - M|y(t_{1}-\tau)|^{q} < 0, 
\end{equation}

but this would imply $y^{'}(t) = 0$, at some $t < t_{1}-\frac{\tau}{2}$, contradicting the fact that $t_{1}$ is the supremum of the earlier set constructed in \eqref{eq:s21}. This entails that, $y$ is increasing on $[0,T^{*})$. Thus $y(t)^{p} - M y(t-\tau)^{q} > 0$, and thus so is $y(t)^{q} - M y(t-\tau)^{q} > 0$. Hence rearranging \eqref{eq:d1}, we have,

\begin{equation}
y^{'}(t) = y(t)^{p} - y(t)^{q} + \left( y(t)^{q} - M y(t-\tau)^{q} \right)
\end{equation}

Using the positivity results derived earlier, blow-up at a finite time $T^{*}$, for sufficiently large data, is immediate in comparison with an ODE of the form
$y^{'}(t) = y(t)^{p} - y(t)^{q}$. 
This proves the theorem.
\end{proof}
\subsection{Proof of Theorem \ref{thm:t13}}
\begin{proof}
\label{proof_thm:t13}
We begin with the following estimate on a time interval $t \in [0,T], T < \infty$.
    \begin{eqnarray}
&& \frac{dx}{dt} \nonumber \\
&\geq& r x -\left( \frac{r}{a k_{min}}  + 1\right) 
h x + \frac{1}{k_{max}}hx^{2} \nonumber \\
&\geq& r x -\left( \frac{r}{a k_{min}}  + 1\right) \frac{1}{a}
h^{2} + \frac{1}{k_{max}}hx^{2} \nonumber \\
&=&r x -\left( \frac{r}{a k_{min}}  + 1\right) \frac{1}{a}
\left(\int^{t}_{0}x(s) ds\right)^{2} + \frac{1}{k_{max}}hx^{2}\nonumber \\
&=& r x - \left( \frac{r}{a k_{min}}  + 1\right) \frac{T^{2}}{a}\left(  x(t-\tau) \right)^{2} + \frac{1}{k_{max}}ax^{3} \nonumber \\
&\geq& \frac{1}{k_{max}}ax^{3} - \left( \frac{r}{a k_{min}}  + 1\right) \frac{T^{2}}{a}\left(  x(t-\tau) \right)^{2}
\end{eqnarray}

Now, if we choose initial data such that.

\begin{equation}
x_{0} >  k_{max} \left( \frac{r}{a k_{min}}  + 1\right)\frac{T^{2}}{a^{2}}
\end{equation}

Then a direct application of Theorem \ref{thm:l1k} yields the finite time blow-up of $x$.

\end{proof}

%\textcolor{red}{Arxiv version: https://arxiv.org/abs/2310.03058} \url{https://arxiv.org/abs/2310.03058}
\bibliography{biblio.bib}

\end{document}